\documentclass[a4paper,11pt]{amsart}
\usepackage[latin1]{inputenc}
\usepackage[T1]{fontenc}
\usepackage[english]{babel}
\usepackage[foot]{amsaddr}
\usepackage{amsfonts, amsmath, amssymb, mathrsfs}
\usepackage{amsthm}
\makeatletter
\@namedef{subjclassname@2020}{%
  \textup{2020} Mathematics Subject Classification}
\makeatother
\usepackage{mathtools}
\usepackage{braket}
\usepackage{graphicx,epstopdf}
\usepackage{tikz}
\usetikzlibrary{arrows,calc,matrix}
\usepackage{soul}
\usepackage{hyperref}
\theoremstyle{plain}
\newtheorem{Theorem}{Thm}[section]
\newtheorem{Thm}[Theorem]{Theorem}
\newtheorem{Lem}[Theorem]{Lemma}
\newtheorem{Cor}[Theorem]{Corollary}
\theoremstyle{definition}
\newtheorem{Rem}[Theorem]{Remark}
\newcommand{\C}{\mathbb{C}}
\newcommand{\E}{\mathbb{E}}
\newcommand{\N}{\mathbb{N}}
\newcommand{\R}{\mathbb{R}}
\newcommand{\Z}{\mathbb{Z}}
\newcommand{\cA}{\mathcal{A}}
\newcommand{\cB}{\mathcal{B}}
\newcommand{\cD}{\mathcal{D}}
\newcommand{\cF}{\mathcal{F}}
\newcommand{\cH}{\mathcal{H}}
\newcommand{\cK}{\mathcal{K}}
\newcommand{\cL}{\mathcal{L}}
\newcommand{\cN}{\mathcal{N}}
\newcommand{\cP}{\mathcal{P}}
\newcommand{\cS}{\mathcal{S}}
\newcommand{\cU}{\mathcal{U}}
\newcommand{\cV}{\mathcal{V}}
\newcommand{\cW}{\mathcal{W}}
\newcommand{\fa}{\mathfrak{a}}
\newcommand{\fc}{\mathfrak{c}}
\newcommand{\fg}{\mathfrak{g}}
\newcommand{\fp}{\mathfrak{p}}
\newcommand{\qu}{\fa_{3,\mathrm{qu}}}
\newcommand{\cl}{\fa_{3,\mathrm{cl}}}
\newcommand{\eps}{\varepsilon}
\newcommand{\ii}{\mathrm{i}}
\newcommand{\GHZ}{\mathrm{GHZ}}
\DeclareMathOperator\diag{diag}
\DeclareMathOperator\supp{supp}
\DeclareMathOperator\Tr{Tr}
\DeclareMathOperator\pos{pos}
\DeclareMathOperator\mar{mar}
\newcommand\scp[1]{\langle #1\rangle}
\begin{document}
\selectlanguage{English}
\title{Quantum marginals, faces, and coatoms}
\author{Stephan Weis$^{1,\ast}$}
\author{Jo\~ao Gouveia$^{2}$}
\address{$^1$Wald-Gymnasium Berlin, 
Germany,
e-mail: \texttt{maths@weis-stephan.de},\newline
ORCID: 0000-0003-1316-9115}
\address{$^2$Department of Mathematics, 
University of Coimbra, 
Portugal,\newline
e-mail: \texttt{jgouveia@mat.uc.pt},
ORCID: 0000-0001-8345-9754}
\address{$^{\ast}$Corresponding author}
\begin{abstract}
The set of quantum marginals is a central object in quantum statistics. 
The faces of this convex set play a decisive role regarding the information 
projection to a hierarchical model, and regarding state tomography from 
marginals. However, the faces of this convex set are widely unexplored. Here, we 
provide an experimental method to explore the maximal faces, the socalled 
coatoms, in the lattice of exposed faces of the set of marginals. The method 
proceeds in three steps: a) sampling extreme points from the dual spectrahedron, 
b) guessing the exact algebraic form of these extreme points, and c) verifying 
the algebraic result. The third step employs ground projectors of local 
Hamiltonians. Using this method, we found a two-parameter family of ground 
projectors of rank five, which support a family of maximal faces of the set of 
two-body marginals of three qubits (the rank is six for three bits). In classical 
information theory, we show that a probability distribution factors with respect 
to an interaction pattern only if it is supported by the ground projector of a 
frustration-free Hamiltonian. We discuss nonexposed points.  
\end{abstract}
%
%
%
%
%
%
%
%
%
%
%
%
%
%
%
%
\maketitle
%
%
%
%
\section{Introduction}
Quantum marginals have been studied in quantum chemistry since the 1960's, see 
for example \cite{Chen-etal2012a,Zeng-etal2019}. The marginals are an
economic representation of a density matrix regarding local interaction patterns. 
As an example, the $k$-body marginals capture all information of a quantum state 
that can be observed through a $k$-local Hamiltonian, an observable that ignores 
interactions which cannot be described by subsystems of $k$ units. 
\par
The set of marginals is fundamental to hierarchical models in statistics
and to the information projection onto such models. 
In classical information theory \cite{Ingarden1976}, these models represent 
patterns of many-body interaction \cite[Section~2.9]{Ay-etal2017} and the 
information projection encodes the maximum-likelihood estimate 
\cite{CsiszarMatus2003}. The information projection can also be used to 
define a complexity measure for a many-body system in terms of the entropy 
distance from a hierarchical model \cite[Section~6.1]{Ay-etal2017}. To 
define the information projection with probability one, one needs to know 
the face of the set of marginals onto which a distributions of nonmaximal 
support projects \cite{CsiszarMatus2003,Kahle2010,Wang-etal2019}.
\par
Faces of marginals are even more important in quantum than in classical
information theory.
The information projection of a probability distribution of nonmaximal 
support can be approximated continuously, but a similar approximation fails 
in quantum mechanics due to the discontinuity of the maximum-entropy inference 
map \cite{WeisKnauf2012}. Hence, the quantum mechanical analogue
\cite{Niekamp-etal2013,Zhou2014,Weis-etal2015,Zeng-etal2019}
of the aforementioned measure of complexity can be computed only for 
commutative subalgebras, for example, stabilizer states \cite{Zhou2014}. 
In any case, finding the maximum-entropy state, and hence the information 
projection, is computationally hard \cite{DiGiorgioMateus2021}.
\par
Prior studies of faces of marginal have appeared in quantum chemistry 
\cite{Erdahl1972,Rosina2000} and quantum state tomography 
\cite{Chen-etal2013,Chen-etal2012a,Chen-etal2012b,
Karuvade-etal2019,Linden-etal2002,Zeng-etal2019}. Still, to the best of our 
knowledge, no faces of the set of two-body mar\-gin\-als of three qubits 
are known besides the extreme points. The extreme points are characterized by the 
fact that each pure state is uniquely determined among all states by its two-body 
marginals \cite{Linden-etal2002}, except for the pure states of the form 
$\alpha\ket{000}+\beta\ket{111}$, the so-called GHZ states. Actually, each of the
uniquely determined pure states is the unique ground state of a two-local 
Hamiltonian \cite{Chen-etal2012b}. The latter implies that the set of marginals 
has no nonexposed points as we will see in Section~\ref{sec:nonexposed}. In 
contrast, using tools from dissipative quantum control theory, recent studies 
show that the set of two-body marginals of six qubits has nonexposed points 
\cite{Karuvade-etal2019}.
\par
We contribute to faces of marginals by describing an experimental 
method. The method allows us to explore the lattice of exposed faces through the 
coatoms (maximal elements). The basic idea is that the coatoms are in a one-to-one 
correspondence with the extreme points of the dual spectrahedron. We sample from 
these extreme points by solving random semidefinite programs numerically. Upon 
guessing the exact algebraic form of an extreme point, or of a family of extreme 
points, linear algebra allows us to test whether a candidate is indeed an extreme 
point, see Remark~\ref{ref:experi}. This test employs the isomorphism from the 
lattice of exposed faces of the set of marginals to the lattice of ground projectors 
of local Hamiltonians \cite{Weis2011,Weis2018b}. Thereby, the coatoms of the 
lattice of exposed faces are in one-to-one correspondence with the coatoms of 
the lattice of ground projectors. All other exposed faces are intersections of
coatoms.
\par
%
\begin{figure}
\includegraphics[width=5cm]{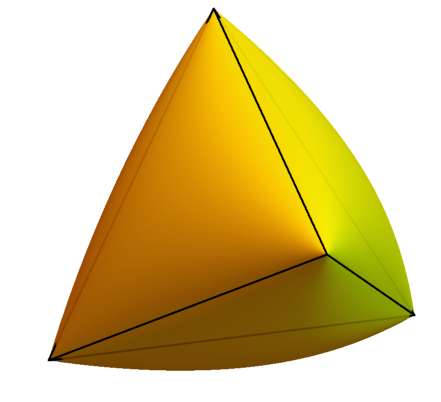}
\caption{\label{fig:cayley}Spectrahedron bounded by the Cayley cubic.}
\end{figure}
%
Our experimental method shows that there are coatoms of rank five in the lattice of 
ground projectors of two-local three-qubit Hamiltonians 
(Section~\ref{sec:rank-three}). This breaks the pattern of
diagonal matrices representing classical information theory \cite{Ingarden1976}, 
where the rank is six. Geometrically, we owe the coatoms of rank five to the 
fact that a spectrahedron of noncommutative matrices typically looks like an 
inflated polyhedron that has high-rank extreme points on its curved boundary. 
An instructive example is the spectrahedron 
\[
\left\{(x,y,z)\in\R^3\mid 
\left(\begin{smallmatrix}1&x&y\\x&1&z\\y&z&1\end{smallmatrix}\right)
\succeq 0\right\},
\]
which contains the one-skeleton of a tetrahedron and which is bounded by the 
Cayley cubic, see Figure~\ref{fig:cayley}. The four vertices of the underlying 
tetrahedron correspond to rank-$1$ matrices, while the rest of the yellow surface 
corresponds to rank-$2$ matrices. The relative interiors of the six edges of the 
tetrahedron constitute the set of boundary points that are not extreme.
\par
We also present results in classical statistics. We describe the faces of the set 
of two-body marginals of three bits in terms of edges of a graph 
(Section~\ref{sec:all3bits}), and we analyze the boundaries of hierarchical models. 
A probability distribution in such a model is equivalently characterized through a 
factorization property and defined as the exponential of a local Hamiltonian. A 
well-known example is the Hammersley-Clifford theorem (Gibbs-Markov equivalence) 
that characterizes the factorization with respect to undirected graphs 
\cite{Lauritzen1996}. The problem of extending the equivalence from maximal to 
nonmaximal support has found a general answer \cite{Geiger-etal2006}. In 
Section~\ref{sec:factor}, we add to this topic by proving that a probability 
distribution factors only if its support set is the ground projector of a 
frustration-free Hamiltonian. 
\par
Returning to quantum mechanics, one reason why the convex geometry of marginals 
is a challenging problem might be the computational complexity of two related 
problems: the local Hamiltonian problem of estimating 
the ground state energy of a local Hamiltonian and the marginal problem
of deciding whether a collection of states is the collection of marginals of a 
global state. The local Hamiltonian problem \cite{Kitaev-etal2002} and the quantum 
marginal problem \cite{Liu2006} are QMA-complete, which means they cannot be solved 
efficiently on a quantum computer. The marginal problem can be solved by a 
hierarchy of semidefinite programs \cite{Yu-etal2021} if the global state is 
assumed to be pure. The quantum marginal problem with non-overlapping subsystems 
is trivial but is QMA-complete for indistinguishable particles, fermions or bosons 
\cite{Liu-etal2007,Wei-etal2010}, where a spectral polytope describes its solution 
\cite{AltunbulakKlyachko2008,Schilling-etal2013,Schilling-etal2020,MaciazekTsanov2017}.
\par
The article is structured as follows. Section~\ref{sec:matrix-algebra1} 
introduces matrix algebras. Section~\ref{sec:vector-spaces} explains the 
experimental method in the general setting of the joint numerical range 
\cite{BonsallDuncan1971}. Section~\ref{sec:many-body} addresses quantum 
marginals. 
\par
The set of two-body marginals of three qubits still offers research challenges. 
It is an open problem to describe the set of $(AB,BC)$-marginals of a three-qubit 
system $ABC$, see \cite{Werner1989} and \cite[Section~4.4.2]{Zeng-etal2019}. The 
analysis of the lattice of faces may be simplified by replacing complex with 
real matrices and by studying ground projectors of frustration-free Hamiltonians 
\cite{Ji-etal2011,Movassagh-etal2010,Zeng-etal2019}. 
\par
%
%
%
\section{Matrix $\ast$-Algebras}
\label{sec:matrix-algebra1}
This section introduces the algebras we employ throughout the paper. 
\par
%
%
%
\subsection{Lattices, State Spaces, Projectors}
\label{sec:lattices-state-spaces}
A \emph{lattice} \cite{Birkhoff1973} is a 
partially ordered set in which any pair of elements has an infimum and a supremum. 
A lattice is \emph{complete} if every subset has an infimum and a supremum. Let 
$\cL$ be a lattice with least element $0$ and greatest element $1$. An 
\emph{atom} of $\cL$ is a minimal element of $\cL\setminus\{0\}$. A \emph{coatom} of 
$\cL$ is a maximal element of $\cL\setminus\{1\}$. The lattice $\cL$ is 
\emph{atomistic} if each of its element is the supremum of a set 
of atoms (such a lattice is called atomic in \cite{Birkhoff1973}). The lattice 
$\cL$ is \emph{coatomistic} if each of its element is the infimum of a set of coatoms.
\par
Let $M_d$ denote the $\ast$-algebra of complex $d\times d$ matrices, and $I_d$ the
$d\times d$ identity matrix. We write the matrix product of $A,B\in M_d$ in the 
form $A.B$ to distinguish it from the tensor product $AB=A\otimes B$ in 
Section~\ref{sec:hierarchical}. We will work with a $\ast$-algebra $\cA\subseteq M_d$
over the reals as this give us the possibility to decrease the dimension (see 
Section~\ref{sec:disk-algebra}). This also includes the $\ast$-algebras over the 
complex field. The \emph{Hilbert-Schmidt inner product} on $\cA$ is 
defined by $\scp{A,B}=\Tr(A^\ast\!.\,B)$ for all $A,B\in\cA$. The real vector space 
of hermitian matrices
\[
\cH(\cA)=\{A\in\cA : A^\ast=A\}
\]
is a Euclidean space with the restricted Hilbert-Schmidt inner product. The 
set $\cH(\cA)$ is partially ordered by the \emph{Loewner order} $A\preceq B$, 
or equivalently $B\succeq A$, which is valid if $B-A$ is positive semidefinite for 
all $A,B\in\cH(\cA)$. We denote the set of positive semidefinite matrices by 
\[
\cA^+=\{A\in\cA:A\succeq 0\}.
\]
The \emph{state space} \cite{AlfsenShultz2001} of the algebra $\cA$ is the set
\[
\cD(\cA)
=\{\rho\in\cA^+: \Tr(\rho)=1\}.
\]
The set $\cA^+$ is a closed, convex cone and $\cD(\cA)$ is a compact, convex set. 
The elements of $\cD(\cA)$ are called \emph{density matrices} or \emph{quantum states} 
\cite{BengtssonZyczkowski2017}. The extreme points of $\cD(\cA)$ are called the 
\emph{pure states} of $\cA$. Endowed with the restricted Loewner order, the set of 
\emph{projectors} in $\cA$,
\[
\cP(\cA)=\{P\in\cA:P=P^\ast=P^2\},
\]
is a complete lattice \cite{AlfsenShultz2001}. 
\par
Rank-one projectors are important, as they are the atoms of $\cP(M_d)$. Every 
rank-one projector $P\in\cP(M_d)$ is a pure state and we write it as 
$P=\ket{\psi}\!\!\bra{\psi}$ in Dirac's notation, where $\ket{\psi}\in\C^d$ is 
any unit vector in the image of $P$. If every atom of $\cP(\cA)$ has rank one, then 
the converse holds: Every pure state of $\cA$ is a rank-one projector.
\par
From now on, we assume the $\ast$-algebra $\cA$ contains the $d\times d$ identity 
matrix $I_d$. In contrast, the multiplicative identity of the $\ast$-algebra
\[
P.\cA.P=\{P.A.P : A\in\cA\}
\]
is $P$ for all projectors $P\in\cP(\cA)$. The assumption of $I_d\in\cA$ guarantees 
that every eigenvalue of a matrix $A\in\cA$ is a spectral value of $A$ in the 
algebra $\cA$, which is important in our definition of a ground projector in 
Section~\ref{sec:geo-state}. 
\par
%
%
\subsection{Diagonal Matrices}
\label{subsec:diagonal-matrices}
Given a finite set $X$, the space $\C^X$ of functions $X\to\C$ is a $\ast$-algebra. 
Let $\delta_x\in\C^X$ be defined by $\delta_x(y)=0$ if $x\neq y$ and $\delta_x(y)=1$ 
if $x=y$, for all $x,y\in X$. The \emph{support} of a function $f\in\C^X$ is the set 
of points $\{x\in X\mid f(x)\neq 0\}$. We identify the set of functions $C(d)\to\C$ 
on the \emph{configuration space} $C(d)=\{0,\ldots,d-1\}$ with the set of $d\times d$ 
diagonal matrices, in such a way that $f\in\C^{C(d)}$ corresponds to the diagonal 
matrix $\,\diag(f(0),f(1),\dots,f(d-1))\,$. In the notation 
Section~\ref{sec:lattices-state-spaces}, the space $\cH(\C^{C(d)})$ of hermitian 
matrices is the set of real functions $C(d)\to\R$, the state space $\cD(\C^{C(d)})$ 
is the simplex of probability distributions on $C(d)$, and the set $\cP(\C^{C(d)})$ 
of projectors is the set of $\{0,1\}$-valued functions on $C(d)$. There is a lattice 
isomorphism
\begin{equation}\label{eq:pro-Cd}
\cP(\C^{C(d)})\to 2^{C(d)},
\quad 
P\mapsto\supp(P),
\end{equation}
from the set of projectors to the power set of the configuration space $C(d)$, which 
maps the rank-one projector $\delta_x$ to $x$ for all $x\in C(d)$.
\par
%
%
%
\subsection{The Qubit-Algebra}
\label{sec:qubitalgebra}
The \emph{qubit} is the information unit of quantum theory. The algebra 
associated with the qubit is the complex $\ast$-algebra of $2\times 2$ matrices 
$M_2$, spanned by the identity matrix
\[
I=I_2=\left(\begin{smallmatrix} 1 & 0 \\ 0 & 1 \end{smallmatrix}\right)
\]
and the \emph{Pauli matrices}
\[
X=\left(\begin{smallmatrix} 0 & 1 \\ 1 & 0 \end{smallmatrix}\right),
\qquad
Y=\left(\begin{smallmatrix} 0 & -\ii \\ \ii & 0 \end{smallmatrix}\right),
\qquad
Z=\left(\begin{smallmatrix} 1 & 0 \\ 0 & -1 \end{smallmatrix}\right).
\]
These matrices also span the real space $\cH(M_2)$ of hermitian matrices.
Any traceless hermitian matrix can be written in the form
\[
\hat{n}\!\cdot\!\vec{\sigma}=n_x X + n_y Y + n_z Z
\]
where $\hat{n}=(n_x,n_y,n_z)\in\R^3$ is the \emph{Bloch vector} and
$\vec{\sigma}=(X,Y,Z)$ the \emph{Pauli vector}. The matrix 
$\hat{n}\!\cdot\!\vec{\sigma}$ has the eigenvalues $\pm|\hat{n}|$ and the 
spectral decomposition
\[
\hat{n}\!\cdot\!\vec{\sigma}=
|\hat{n}|(I+\tfrac{\hat{n}}{|\hat{n}|}\!\cdot\!\vec{\sigma})/2
-|\hat{n}|(I-\tfrac{\hat{n}}{|\hat{n}|}\!\cdot\!\vec{\sigma})/2,
\qquad \hat{n}\neq 0.
\]
The state space of $M_2$ is the \emph{Bloch ball}
\[
\cD(M_2)=\left\{(I+\hat{n}\!\cdot\!\vec{\sigma})/2 : 
 \hat{n}\in\R^3, |\hat{n}|\leq 1 \right\}.
\]
The set of pure states is the \emph{Bloch sphere}
$\{(I+\hat{n}\!\cdot\!\vec{\sigma})/2 : \hat{n}\in\R^3, |\hat{n}|=1\}$.
\par
%
%
%
\subsection{The Disk-Algebra}
\label{sec:disk-algebra}
The real $\ast$-algebra $M_2(\R)=\mathrm{span}_\R\{I,X,Z,\ii Y\}$ is interesting 
as it is noncommutative and has a smaller dimension than $M_2$. The 
space of hermitian matrices is 
\[
\cH(M_2(\R))=\mathrm{span}_\R\{I,X,Z\}.
\]
The state space is the disk
$\cD(M_2(\R))=\{\rho\in \cD(M_2) : \scp{\rho,Y}=0 \}$, a cross section of the 
Bloch ball.
\par
%
%
%
\subsection{The Bit-Algebra}
\label{sec:bitalgebra}
The information unit of digital computers is the \emph{bit}, which has the 
configuration space $C(2)=\{0,1\}$. Thinking of the elements of $\C^{C(2)}$ as 
$2$-by-$2$ diagonal matrices as in Section~\ref{subsec:diagonal-matrices}, we 
write $\C^{C(2)}$ as the span of the identity matrix $I$ and the Pauli matrix 
$Z$ introduced in Section~\ref{sec:qubitalgebra}. As per the lattice isomorphism 
\eqref{eq:pro-Cd}, the rank-one projectors 
\[
\tfrac{1}{2}\big(I+(-1)^x Z\big),
\qquad
x=0,1,
\]
of $\C^{C(2)}$ are in a one-to-one correspondence with the configurations 
$0$ and $1$.
\par
%
%
%
\section{Lattices Associated with a Space of Hermitian Matrices}
\label{sec:vector-spaces}
%
%
The projection of the state space $\cD(\cA)$ of the $\ast$-algebra $\cA$ onto 
a space of hermitian matrices is the joint numerical range (up to a linear 
isomorphism), which we denote by $\cW$. We discuss the exposed faces and normal 
cones of $\cW$ and of its dual spectrahedron. We use the lattice isomorphisms 
of Figure~\ref{fig:CD} to describe the experimental approach to the coatoms 
of the lattice of exposed faces of $\cW$ in Remark~\ref{ref:experi} at the end 
of the section. We refer to \cite{Rockafellar1970,Schneider2014} regarding 
convex geometry, and to \cite{AubrunSzarek2017} regarding the convex geometry 
of quantum states.
\par
%
%
%
\subsection{Exposed Faces and Normal Cones}
\label{sec:expo-norm}
Let $(\E,\scp{\cdot,\cdot})$ be a Euclidean space and $C\subseteq\E$ a convex 
subset. An \emph{exposed face} of $C$ is a subset of $C$, which is either empty or 
equal to the set of points at which a linear function attains its minimum on 
$C$. We denote the set of exposed faces of $C$ by $\cF(C)$. If $C$ is compact then 
the minimum $\mu_{C,u}=\min_{x\in C}\scp{x,u}$ exists for all $u\in\E$ and we define 
the map
\[
F_C:\E\to\cF(C), 
\qquad
u\mapsto\{x\in C: \scp{x,u}=\mu_{C,u}\}.
\]
We call $F_C(u)$ the exposed face of $C$ \emph{exposed} by the vector $u$. A point 
$x\in C$ is an \emph{exposed point} if $\{x\}$ is an exposed face. Partially ordered 
by inclusion, the set $\cF(C)$ is a complete lattice, and the infimum is the 
intersection. 
\par
The \emph{normal cone} to $C$ at a point $x\in C$ is the closed convex cone
\[
N_C(x)=\{u\in\E \mid \scp{y-x,u} \geq 0 \; \forall y\in C \}.
\]
The \emph{normal cone} to $C$ at a nonempty convex subset $G\subseteq C$ is 
defined as the intersection $N_C(G)=\cap_{x\in G}N_C(x)$. We put 
$N_C(\emptyset)=\E$. Partially ordered by inclusion, the set $\cN(C)$ of normal 
cones to $C$ is a complete lattice, and the infimum is the intersection. 
\par
In a slight abuse of the symbol $N_C$, we define the map
\begin{equation}\label{equ:anti-F-N}
N_C:\cF(C)\to\cN(C),\qquad
F\mapsto N_C(F).
\end{equation}
If $C$ is not a singleton, then this map is an antitone lattice isomorphism. The 
statements of this section are proved in \cite{Weis2012}.
\par
%
%
%
\subsection{Convex Duality}
\label{sec:conv-dual}
Let $(\E,\scp{\cdot,\cdot})$ be a Euclidean space and denote the orthogonal 
projection onto a subspace $U\subseteq\E$ by $\pi_U:\E\to\E$. The 
\emph{dual convex cone} to a subset $C\subseteq\E$ is the closed convex cone
\[
C^\vee=\{u\in\E \mid \scp{u,x}\geq 0 \; \forall x\in C \}.
\]
If $C$ is a closed convex cone, then $C=(C^\vee)^\vee$ holds. If $C=C^\vee$, 
then $C$ is called a \emph{self-dual convex cone}. The \emph{dual convex set} to 
any subset $C\subseteq\E$ is
\[
C^\circ=\{u\in\E \mid 1 + \scp{u,x}\geq 0 \; \forall x\in C \}.
\]
The set $C^\circ$ is a closed convex set containing the origin. If $C$ 
is a compact, convex set containing the origin as an interior point, then the dual 
convex set $C^\circ$ is compact and contains the origin as an interior point, too 
\cite{Rockafellar1970,Schneider2014}. 
\par
Section~\ref{sec:spectrahedra} uses the following one-to-one correspondence between 
normal cones of $C$ and exposed faces of $C^\circ$. If $C$ is a compact, convex set 
containing the origin as an interior point and if $\dim(\E)\geq 1$, then the map
\begin{equation}\label{eq:dual-convex}
\chi_C:\cN(C)\to\cF(C^\circ),
\qquad
N\mapsto
\left\{
\begin{array}{ll}
N\cap\partial C^\circ & \text{if $N\neq \E$},\\
C^\circ               & \text{if $N=\E$},
\end{array}
\right.
\end{equation}
is an isotone lattice isomorphism, where $\partial C^\circ$ is the boundary of 
$C^\circ$. The composition of the maps \eqref{equ:anti-F-N} and 
\eqref{eq:dual-convex} is the antitone lattice isomorphism $\cF(C)\to\cF(C^\circ)$ 
that maps an exposed face to its \emph{conjugate face} \cite{Schneider2014}. The 
inverse isomorphism to \eqref{eq:dual-convex} is
\begin{equation}\label{eq:dual-convex-inverse}
\chi_C^{-1}:\cF(C^\circ)\to\cN(C),
\qquad
F\mapsto\pos(F),
\end{equation}
where $\pos(F)=\{\lambda x : \lambda\geq 0, x\in F\}$ is the positive hull of 
any nonempty exposed face $F$ of $C^\circ$, and $\pos(\emptyset)=\{0\}$, see for 
example \cite[Section 8]{Weis2012}. Lemma~7.2 of \cite{Weis2012} shows that the 
cone $\chi_C^{-1}(F)$ is the normal cone to $C$ at the exposed face $F_C(u)$ of 
$C$ which is exposed by any nonzero vector $u$ in the relative interior of 
$\chi_C^{-1}(F)$, for all exposed faces $F\neq C^\circ$ of $C^\circ$.
\par
The following construction is fundamental in Section~\ref{sec:spectrahedra}.
Let $C\subseteq\E$ be a closed convex cone with interior point $\epsilon\neq0$. Then 
\[
B_{C,\epsilon}
=\{u\in C^\vee : \scp{u,\eps} = 1\}
\]
is a compact convex set, which is a \emph{base} of $C^\vee$. Let $U\subseteq\E$ be a 
linear subspace incident with $\epsilon$, let 
\[
V_{U,\epsilon}=\{u\in U:\scp{u,\eps} = 0\}
\]
be the orthogonal complement to $\epsilon$ in $U$, and 
\[
S_{C,U,\epsilon} 
=\{x\in V_{U,\epsilon} : \eps + x \in C \}
\]
an affine section of the cone $C$.
\par
\begin{Lem}\label{lem:convex-dual}
Let $C\subseteq\E$ be a closed convex cone and let $\epsilon\neq 0$ be an interior 
point of $C$. Let $U\subseteq\E$ be a linear subspace incident with $\epsilon$. Then 
$S_{C,U,\epsilon}$ is the dual convex set to $\pi_{V_{U,\epsilon}}(B_{C,\epsilon})$
with respect to the Euclidean space $V_{U,\epsilon}$.
\end{Lem}
\begin{proof}
Let $x\in V_{U,\epsilon}$ and let $B=B_{C,\epsilon}$. Then
\begin{align*}
x\in S_{C,U,\epsilon} &\iff \eps + x \in C 
 \iff \forall u\in C^\vee : 
 \scp{u,\eps + x} \geq 0\\
 & \iff \forall u\in B : 
 \scp{u,\eps + x} \geq 0
 \iff \forall u\in B :  
 1 + \scp{u,x} \geq 0\\
 & \iff \forall u\in \pi_{V_{U,\epsilon}}(B) : 
 1 + \scp{u,x} \geq 0.
\end{align*}
\nopagebreak%
This proves the claim.
\end{proof}
%
%
\subsection{Geometry of the State Space}
\label{sec:geo-state}
The exposed faces and the normal cones of the state space $\cD(\cA)$ are 
represented in terms of projectors.
\par
Let $P_0:\cH(\cA)\to\cP(\cA)$ denote the map from the set of hermitian matrices to the
set of projectors, where $P_0(A)$ is the spectral projector of $A$ corresponding to the 
smallest eigenvalue of $A$. We call $P_0(A)$ the \emph{ground projector} of $A$ by its
name in physics if $A$ represents an energy observable. 
\par
The exposed face of the state space $\cD(\cA)$ exposed by $A\in\cH(\cA)$ is
\[
F_{\cD(\cA)}(A)
=\cD(P_0(A).\cA.P_0(A)).
\]
Note that
\[
F_{\cD(\cA)}(A)
=\{\rho\in\cD(\cA) \mid S(\rho)\preceq P_0(A) \},
\]
where $S(\rho)$ is the \emph{support projector} of $\rho$, the sum of the spectral 
projectors corresponding to the nonzero eigenvalues. Moreover, the map
\begin{equation}\label{eq:iso-P-FD} 
\phi_\cA:\cP(\cA)\to\cF(\cD(\cA)),\qquad
P\mapsto\cD(P.\cA.P)
\end{equation}
is an isotone lattice isomorphism from the lattice of projectors $\cP(\cA)$ to the 
lattice of exposed faces of $\cD(\cA)$, see for example \cite{AlfsenShultz2001} or 
\cite[Section~2.3]{Weis2011}. As 
\begin{equation}\label{eq:factHF}
F_{\cD(\cA)}(A)=\phi_\cA\circ P_0(A),
\qquad
A\in\cH(\cA),
\end{equation} 
the map $F_{\cD(\cA)}:\cH(\cA)\to\cF(\cD(\cA))$ factors through $\cP(\cA)$.
\par
The concatenation of the maps \eqref{eq:iso-P-FD} and \eqref{equ:anti-F-N} 
is the antitone lattice isomorphism
\begin{equation}\label{eq:iso-P-N}
\nu_\cA:\cP(\cA)\to\cN(\cD(\cA)),\qquad
P\mapsto N_{\cD(\cA)}\circ\phi_\cA(P),
\end{equation}
where 
\[
\nu_\cA(P)= \{A\in\cH(\cA) \mid P\preceq P_0(A)\}
\]
is the normal cone to $\cD(\cA)$ at the exposed face $\phi_\cA(P)$, see 
\cite[Section~3]{Weis2018b}. We ignore the case $\cA\cong\C$ where 
$\cD(\cA)=\{I_d/d\}$ and \eqref{eq:iso-P-N} is not injective.
\par
%
%
\subsection{The Joint Numerical Range and its Exposed Faces}
\label{sec:jnrs}
In the sequel, let $\cU\subseteq\cH(\cA)$ be a vector space of hermitian 
matrices, and let $\pi_\cU:\cH(\cA)\to\cH(\cA)$ denote the orthogonal 
projection onto $\cU$.  
\par
If $F_1,\ldots,F_k$ is a spanning 
set of $\cU$, then the map $\E:\cH(\cA)\to\R^k$, $A\mapsto\scp{A,F_i}_{i=1}^k$, 
factors through $\cU$ as per $\E=\E\circ\pi_\cU$. The map 
$\cU\stackrel{\E}{\rightarrow}\E(\cU)$ is a linear isomorphism, see 
\cite[Remark~1.1]{Weis2011}, which restricts to the bijection  
\[
\pi_\cU\left(\cD(\cA)\right)\stackrel{\E}{\longrightarrow}\E\left(\cD(\cA)\right).
\]
The set $\E\left(\cD(\cA)\right)$ is known as the \emph{joint numerical range} 
\cite{BonsallDuncan1971} of $F_1,\ldots,F_k$. Here we call the set 
$\pi_\cU\left(\cD(\cA)\right)$ \emph{joint numerical range} of $\cU$.
\par
Equation \eqref{eq:factHF} shows that the function which maps a hermitian matrix $A$ 
to the exposed face of $\cD(\cA)$ exposed by $A$ factors through the lattice of 
projectors. If $A\in\cU$ then the map factors also through the lattice of 
exposed faces of the joint numerical range, 
\begin{equation}\label{eq:isoP0FW-prop}
F_{\cD(\cA)}(A)
=\phi_\cA\circ P_0(A)
=\pi_\cU|_{\cD(\cA)}^{-1}\circ\pi_\cU\circ\phi_\cA\circ P_0(A),
\qquad
A\in\cU.
\end{equation}
As detailed in Section~3.1 of \cite{Weis2011}, by endowing the set of ground 
projectors 
\[
\cP_0(\cU) = \{ P_0(A) : A\in\cU\} \cup \{0\}
\]
with the Loewner order and the set of exposed faces 
$F_{\cD(\cA)}(\cU)\cup\{\emptyset\}$ with the partial order of inclusion, one 
obtains the lattice isomorphisms
\begin{equation}\label{eq:isoP0FW}
\begin{tikzpicture}[baseline=(m-1-1.base)]
\matrix (m) [matrix of math nodes]{%
\cP_0(\cU) 
 & [3.2em] \phi_\cA(\cP_0(\cU))=F_{\cD(\cA)}(\cU)\cup\{\emptyset\}
 & [3.2em] \cF(\pi_\cU(\cD(\cA))).\\};
\path[-stealth](m-1-1) edge node[above]{$\phi_\cA$} (m-1-2);
\path[-stealth]($(m-1-2.east)+(0,.1)$) 
edge node[above]{$\pi_\cU$} ($(m-1-3.west)+(0,.1)$);
\path[-stealth]($(m-1-3.west)-(0,.1)$) 
edge node[below]{$\pi_\cU|_{\cD(\cA)}^{-1}$} ($(m-1-2.east)-(0,.1)$);
\end{tikzpicture}
\end{equation}
\par
The lattices $\cP_0(\cU)$, $\phi_\cA(\cP_0(\cU))$, and $\cF(\pi_\cU(\cD(\cA)))$ 
are complete, co\-at\-om\-is\-tic lattices \cite[Corollary 6.5]{Weis2018a}. The 
infimum in the lattices $\phi_\cA(\cP_0(\cU))$ and $\cF(\pi_\cU(\cD(\cA)))$ is the 
intersection. The infimum in $\cP_0(\cU)$ is the same as the infimum in the 
lattice $\cP(\cA)$ of all projectors, restricted to subsets of $\cP_0(\cU)$, of 
course \cite[Section 4]{Weis2018b}.
\par
%
%
\subsection{Normal Cones of the Joint Numerical Range}
\label{sec:normal-cones-JNR}
We discuss the antitone isomorphism between the ground projectors and the 
normal cones of the joint numerical range. The atoms (rays) of the lattice 
of normal cones characterize the coatoms of the lattice of ground projectors.
\par
If $\pi_{\cU}(\cD(\cA))$ is not a singleton, then the lattice isomorphisms 
\eqref{eq:isoP0FW} and \eqref{equ:anti-F-N} concatenate to the antitone lattice 
isomorphism
\begin{equation}\label{eq:iso-PU-NU}
\cP_0(\cU)\to\cN(\pi_\cU(\cD(\cA))),\qquad
P\mapsto \nu_\cA(P)\cap\cU.
\end{equation}
Here, $\nu_\cA(P)$ is a normal cone to the state space, see 
Equation~\eqref{eq:iso-P-N}, and
\[
\nu_\cA(P)\cap \cU = \{A\in\cU | P\preceq P_0(A)\}
\]
is the normal cone to the joint numerical range $\pi_\cU(\cD(\cA))$ at the convex 
subset $\pi_\cU\circ\phi_\cA(P)$ for all $P\in\cP(\cA)$. 
See \cite[Section~4]{Weis2018b} for details.
\par
From now on we assume that the space of hermitian matrices $\cU\subseteq\cH(\cA)$ 
contains the $d\times d$ identity matrix $I_d$. A somewhat simpler object than the 
normal cone $\nu_\cA(P)\cap\cU$ is the cone
\begin{align}\label{eq:KP}
\cK(P)
&= P'\!.\cA^+\!\!.P'\cap\cU \\
&= \{A\in\cU \mid A\succeq 0, P\preceq\ker(A)\},\nonumber
\qquad
P\in\cP(\cA).
\end{align}
Here, $P'=I_d-P$ denotes the complementary projector to $P$. 
\par
\begin{Lem}[Theorem~5.1 of \cite{Weis2018b}]\label{lem:charPU}
Let $\cU\subseteq\cH(\cA)$ be a linear subspace with $I_d\in\cU$ and let $P\in\cP(\cA)$. 
Then $P$ lies in $\cP_0(\cU)$ if and only if $P$ is the greatest element of the set 
of all $Q\in\cP(\cA)$ which satisfy $\cK(Q)=\cK(P)$.
\end{Lem}
Lemma~\ref{lem:charPU} yields a necessary condition for projectors to lie in
$\cP_0(\cU)$.
\par
\begin{Lem}\label{lem:Kp=0}
Let $\cU\subseteq\cH(\cA)$ be a linear subspace with $I_d\in\cU$ and let $P\neq I_d$
be a projector in $\cA$. If there exists a hermitian matrix $A\in\cH(\cA)$ 
orthogonal to $\cU$ and a nonzero number $\lambda\neq 0$ such that 
$P'\!.A.P'=\lambda P'\!$, then $P\not\in\cP_0(\cU)$.
\end{Lem}
\begin{proof}
Let $A\in\cU^\perp$ and $\lambda\neq 0$ such that $P'\!.A.P'=\lambda P'$ and let 
$U\in\cK(P)$. Since $U=P'\!.U.P'$, we have
\[
\lambda\Tr(U)
=\lambda \scp{P',U}
=\scp{P'\!.A.P',U}
=\scp{A,U}
=0.
\]
Since $\lambda\neq 0$ we get $\Tr(U)=0$. As $U\succeq 0$ this implies $U=0$.
The claim follows from Lemma~\ref{lem:charPU} as $\cK(I_d)=\{0\}$.
\end{proof}
The second claim of the following lemma is clear as the cone of positive 
semidefinite matrices $\cA^+$ contains no lines.
\par
\begin{Lem}[Theorem~6.1 of \cite{Weis2018b}]\label{lem:ray-coatoms}
Let $\cU\subseteq\cH(\cA)$ be a linear subspace with $I_d\in\cU$ and let 
$P\in\cP_0(\cU)$. Then $P$ is a coatom of $\cP_0(\cU)$ if and only if $\cK(P)$ 
is a ray. This happens if and only if $\dim\cK(P)=1$. 
\end{Lem}
Finding the dimension of the cone $\cK(P)$ is a problem of linear algebra.
\par
\begin{Lem}\label{lem:dimKp}
Let $\cU\subseteq\cH(\cA)$ be a linear subspace with $I_d\in\cU$ and let 
$P\in\cP_0(\cU)$. Then the real span of the cone $\cK(P)$ is 
$\cH(P'\!.\cA.P')\cap\cU$.
\end{Lem}
\begin{proof}
The cone $\cK(0)=\cA^+\cap\cU$ has the span $\cU$ as required, as $I_d\in\cU$.
Let $P\neq 0$. As 
$P\in\cP_0(\cU)$ and as $I_d\in\cU$, there is a (positive semidefinite) matrix 
$U\in\cU$ such that $P=P_0(U)$ and $U.P=0$. Hence, $U$ is invertible in the algebra 
$P'\!.\cA.P'$. Thus, $U$ is an interior point of the cone of positive semidefinite 
matrices $P'\!.\cA^+\!\!.P'$ with respect to the topology of $\cH(P'\!.\cA.P')$, 
see Prop.~2.7 of \cite{Weis2011}. This proves the claim.
\end{proof}
%
%
\subsection{Finding Coatoms \emph{via} Semidefinite Programming}
\label{sec:spectrahedra}
We show that the coatoms of the lattice of ground projectors $\cP_0(\cU)$ are in a 
one-to-one correspondence with the extreme points of a spectrahedron. This yields
a numerical algorithm to find candidates for coatoms, and an algebraic method to 
verify the candidates are indeed coatoms.
\par
Besides the hypothesis that $I_d\in\cU$, we assume $\dim(\cU)\geq 2$ from now on.
We introduce the space
\[
\cV
=\{A\in\cU : \scp{A,I_d} = 0 \}
=\{A\in\cU : \Tr(A)=0 \}.
\]
The joint numerical ranges $\pi_{\cV}(\cD(\cA))=\pi_{\cU}(\cD(\cA))-I_d/d$ are 
translates of each other, and the lattices of ground projectors 
$\cP_0(\cV)=\cP_0(\cU)$ coincide. The affine section
\[
\cS(\cU)
=\{A\in\cV : I_d + A\in\cA^+\}
\]
of the cone of positive semidefinite matrices is a \emph{spectrahedron} 
\cite{RamanaGoldman1995}
\par
It is well known that the cone of positive semidefinite matrices $(M_d)^+$ is a 
self-dual convex cone within the Euclidean space of hermitian matrices $\cH(M_d)$. 
The analogue is true for every real $\ast$-algebra $\cA\subseteq M_d$, see 
Corollary~2.8 of \cite{Weis2011}. Therefore, Lemma~\ref{lem:convex-dual} shows 
\begin{equation}\label{eq:dualJNR}
\cS(\cU)=\cW^\circ\!\!,
\end{equation}
where $\cW=\pi_{\cV}(\cD(\cA))$ denotes the joint numerical range. That is to say, 
the spectrahedron $\cS(\cU)$ is the dual convex set to $\cW$.
\par
Combining two lattice isomorphisms, we identify $\cP_0(\cU)$ and the set of 
exposed faces of $\cS(\cU)$. Equation~\eqref{eq:iso-PU-NU} provides an antitone 
lattice isomorphism $\cP_0(\cU)\to\cN(\cW)$ to the lattice of normal cones of 
$\cW$, as $\cW$ is not a singleton under the chosen assumptions. As $\cW$ is 
compact, Equation~\eqref{eq:dual-convex} provides the isomorphism 
$\cN(\cW)\to\cF(\cS(\cU))$. The function composition of \eqref{eq:iso-PU-NU} 
and \eqref{eq:dual-convex} is the antitone lattice isomorphism 
\begin{equation}\label{eq:iso-PU-FS}
\iota:\cP_0(\cU)\to\cF(\cS(\cU)),
\qquad
P\mapsto
\left\{
\begin{array}{ll}
\nu_\cA(P)\cap\partial\cS(\cU) & \text{if $P\neq 0$},\\
\cS(\cU)                       & \text{if $P=0$},
\end{array}
\right.
\end{equation}
where 
$\nu_\cA(P)\cap\cV$ is a normal cone to $\cW$, as introduced in 
Equation~\eqref{eq:iso-PU-NU}.
\par
\begin{figure}
\begin{tikzpicture}[>=stealth,->,shorten >=2pt]
\matrix[matrix of math nodes,
column sep=4em,row sep=1.5em,inner sep=3pt](m){
 \,  & \cP_0(\cU)    & & \,            \\
 \,  & \,            & & \cF(\cS(\cU)) \\
 \cV & \cF(\cD(\cA)) & & \,            \\
 \,  & \,            & & \cN(\cW)      \\
 \,  & \cF(\cW)      & & \,            \\};
\draw(m-3-1) to node[midway, above, inner sep=8pt]{\small $P_0$} (m-1-2);  
\draw(m-3-1) to node[near end, above]{\small $F_{\cD(\cA)}$} (m-3-2);
\draw(m-3-1) to node[midway, below, inner sep=8pt]{\small $F_\cW$} (m-5-2);
\draw(m-1-2) to node[midway, right]{\small $\phi_\cA$} (m-3-2);  
\draw(m-3-2) to node[midway, right]{\small $\pi_\cV$} (m-5-2);  
\draw(m-1-2) to node[midway, above, inner sep=5pt]{\small $\iota$} (m-2-4);  
\draw(m-1-2) to node[near end, left, inner sep=10pt]{\small $\nu_\cA(\,\cdot\,)\cap\cV$} (m-4-4);  
\draw(m-5-2) to node[midway, below, inner sep=5pt]{\small $N_\cW$} (m-4-4);  
\draw[transform canvas={xshift=-2pt}](m-4-4) to[xshift=-1] node[midway, left]{\small $\chi_\cW$} (m-2-4);
\draw[transform canvas={xshift=2pt}](m-2-4) to node[midway, right]{\small $\pos$} (m-4-4);
\end{tikzpicture}
\caption{\label{fig:CD}Commutative diagram with isomorphisms between the lattices 
$\cP_0(\cU)$, $\cF(\cW)$, $\cN(\cW)$, and $\cF(\cS(\cU))$.}
\end{figure}
We invert the isomorphism \eqref{eq:iso-PU-FS}. Note that all faces of the 
spectrahedron $\cS(\cU)$ are exposed faces \cite{RamanaGoldman1995}. In particular, 
all extreme points are exposed points.
\par
\begin{Thm}\label{thm:coatoms-exposed}
Let $\cU\subseteq\cH(\cA)$ be a subspace with $\dim(\cU)\geq 2$ and $I_d\in\cU$. Let 
$(P,F)\neq (0,\cS(\cU))$ be a point in the graph of the isomorphism $\iota$. Then 
$P=P_0(A)$ holds for any nonzero matrix $A$ in the relative interior of the positive
hull $\pos(F)$. The map $\iota$ restricts to the bijection
\begin{equation}\label{eq:coatoms}
\Big\{\text{coatoms of $\cP_0(\cU)$}\Big\}
\rightarrow
\Big\{\{A\} \mid \text{$A$ is an exposed point of $\cS(\cU)$}\Big\}.
\end{equation}
Let $(P,\{A\})$ be a point in the graph of the map \eqref{eq:coatoms}. Then $A$ is 
the unique matrix in $\cV$ with ground projector $P=P_0(A)$ for which $I_d+A$ is 
positive semidefinite of nonmaximal rank. If $P\in\cU$ then 
$A=\tfrac{\Tr(P)}{\Tr(P')}P'-P$.
\end{Thm}
\begin{proof}
The commutative diagram in Figure~\ref{fig:CD} provides an overview of the maps
introduced in Section~\ref{sec:vector-spaces}, which are relevant to this proof. 
By applying the positive hull operator to the equation $\iota(P)=F$, we obtain 
\[
\nu_\cA(P)\cap\cV
=\pos(F).
\]
Let $A$ be a nonzero point in the relative interior of $\pos(F)$. As discussed 
below of Equation~\eqref{eq:dual-convex-inverse}, the convex cone $\pos(F)$ is 
the normal cone to $\cW$ at the exposed face $F_\cW(A)$, that is to say
\[
\nu_\cA(P)\cap\cV
=N_\cW\circ F_\cW(A).
\]
The commutative diagram then shows $P=P_0(A)$.
\par
The isomorphism $\iota$ restricts to the bijection~\eqref{eq:coatoms}, since 
every atom of the lattice of exposed faces of $\cS(\cU)$ is an exposed point. 
To prove this, it suffices to show that every nonempty exposed face $F$ of 
$\cS(\cU)$ contains an exposed point of $\cS(\cU)$. Since $F$ is compact, it 
has an extreme point $A$ by Minkowski's theorem \cite{Schneider2014}. As $F$ 
is a face of $\cS(\cU)$, the point $A$ is an extreme point of $\cS(\cU)$, and 
hence an exposed point of $\cS(\cU)$.
\par
Let $P$ be a coatom of $\cP_0(\cU)$ and let $A$ be an exposed point of $\cS(\cU)$ 
such that $\{A\}=\iota(P)$. Since $A$ is in the 
relative interior of the ray $\pos(\{A\})$, we get $P=P_0(A)$ as above. As $P$ 
is a coatom of $\cP_0(\cU)$, the relation $P\preceq P_0(B)$ implies $P=P_0(B)$ for 
all nonzero traceless matrices $B\in\cU$. Hence, the ray 
\[
\nu_\cA(P)\cap\cV=\pos(\{A\})
\]
consists of all matrices $B\in\cV$ such that $P=P_0(B)$, and of zero, as per the 
definition of $\nu_\cA(P)$ in Equation~\eqref{eq:iso-P-N}. The ray intersects the 
boundary $\partial\cS(\cU)$ only in $A$. This completes the characterization of 
$A$, because boundary points $B$ of the spectrahedron $\cS(\cU)$ are characterized 
by $I_d+B$ being positive semidefinite of nonmaximal rank. If $P\in\cU$, then the 
matrix $\frac{\Tr(P)}{\Tr(P')}P'-P$ fulfills the characterizing conditions of $A$.
\end{proof}
Theorem~\ref{thm:coatoms-exposed} underpins the initial idea to this 
article.
\par
\begin{Rem}[Experimental Search for Coatoms]\label{ref:experi}
The coatoms of the lattice $\cF(\cW)$ of exposed faces of the joint numerical
range $\cW=\pi_\cV(\cD(\cA))$ are amenable to a numerical exploration, supported 
by convex geometry and linear algebra.
\par
The map from the exposed faces of $\cW$ to their conjugate faces defines an 
antitone isomorphism $\cF(\cW)\to\cF(\cS(\cU))$ to the lattice of exposed faces of 
the spectrahedron $\cS(\cU)$. This map induces a one-to-one correspondence between 
the coatoms of $\cF(\cW)$ and the atoms of $\cF(\cS(\cU))$, which are the extreme 
points of $\cS(\cU)$. Numerically, one can draw linear functionals from the dual 
space $\cV^\ast$ of $\cV$ at random, and minimize them on $\cS(\cU)$. The minimum 
of a generic linear functional is attained at a single extreme point of $\cS(\cU)$. 
This means that this random search will allow us to sample extreme points, or at 
least numerical approximations of such. The minimization can be done efficiently, 
using semidefinite programming \cite{Ben-TalNemirovski2001}.
\par
To illustrate some subtleties of the underlying process let us revisit the Cayley 
cubic example of Figure~\ref{fig:cayley}. In that case there are two types of 
extreme points, the four rank one vertices and the surface of rank two extreme 
points. While there are only four rank one extreme points, their normal cones have 
a high volume, hence it is quite likely that while searching in a random direction 
we end up sampling those points. In Figure~\ref{fig:cayleynormal} one can see the 
possible search directions in $\mathbb{R}^3$, color coded by which type of extreme 
point they lead to. The directions that lead to rank one matrices form four equal 
spherical caps pairwise tangent. One can easily calculate that a random search 
would therefore lead to a rank one matrix around $84.5\%$ of the times, and a rank 
two matrix otherwise. The exceptional directions that would lead to linear forms 
that are minimized in higher dimensional faces are the six tangency points of the 
caps.
\par
This gives us some hint of possible issues if one wants to find representatives for 
all classes of extreme points. In high dimensions, if the union of the normal cones
of the extreme points in some class is of very low volume, it might be hard to 
sample by a uniformly generated random search direction. This problem should not be 
as acute in moderate dimensions, and does not stop us from attempting to find new 
interesting classes of extreme points.
\par
\begin{figure}
\includegraphics[width=5cm]{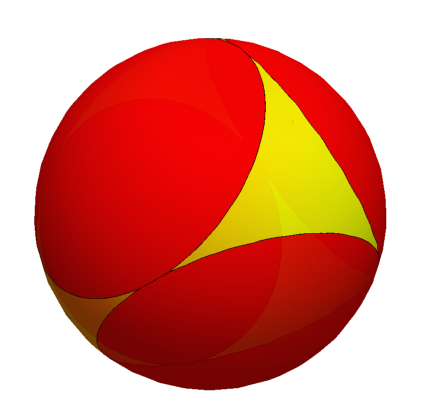}
\caption{\label{fig:cayleynormal}Search directions in the Cayley cubic.}
\end{figure}


\par
The lattice $\cF(\cW)$ is also isomorphic to the lattice $\cP_0(\cU)$ of ground
projectors. This brings about a one-to-one correspondence between the coatoms of 
$\cF(\cW)$ and the coatoms of $\cP_0(\cU)$. Let $A\in\cV$ be an arbitrary matrix, 
for example an output of the random search described above. Then the ground 
projector $P=P_0(A)$ lies in the lattice $\cP_0(\cU)$. Lemma~\ref{lem:ray-coatoms} 
and~\ref{lem:dimKp} above prove that $P$ is a coatom of $\cP_0(\cU)$ if and only 
if the real vector space 
\begin{equation}\label{eq:experimental-app}
\cH(P'\!.\cA.P')\cap\cU
\end{equation} 
is a line. Here, $P'=I_d-P$ is the complementary projector to $P$. Verifying that 
the vector space in equation \eqref{eq:experimental-app} has dimension one allows 
us to confirm that $P=P_0(A)$ is a coatom.
\end{Rem}
It is worth stating some geometric intuition behind 
Theorem~\ref{thm:coatoms-exposed}. Let $A\in\cV$, $P=P_0(A)$, and let 
the space $\cH(P'\!.\cA.P')\cap\cU$ be a line. Then the ray 
$\{\lambda A\mid\lambda\geq 0\}$ meets an extreme point $\lambda_0A$
of the spectrahedron $\cS(\cU)$, for a unique $\lambda_0>0$. In this case,
$I_d+\lambda_0A$ generates the ray $\cK(P)=P'\!.\cA^+\!\!.P'\cap\cU$ defined 
above in Equation~\eqref{eq:KP}.
\par
%
%
\section{Quantum Marginals}
\label{sec:many-body}
We discuss the experimental method regarding coatoms of the lattice of exposed 
faces of the convex set of quantum marginals. We examine the case of three bits 
in Section~\ref{sec:all3bits} as a reference for qubits. The
Sections~\ref{sec:first-glimpse} and~\ref{sec:rank-three} examine three qubits.  
Section~\ref{sec:factor} is an excursion to probability distribution that factor. 
Section~\ref{sec:nonexposed} is an excursion to tomography and nonexposed faces.
\par
%
%
\subsection{Marginals and Local Hamiltonians}
\label{sec:hierarchical}
We specify an interaction pattern on a many-body system of $N\in\N$ units by
choosing a family $\fg$ of subsets  of $\Omega=\{1,2,\ldots,N\}$. Let 
$(d_1,d_2,\ldots,d_N)$ be a sequence of natural numbers and 
$\fa=(\cA_1,\cA_2,\ldots,\cA_N)$ a sequence of $\ast$-algebras, where $\cA_i$ 
is included in $M_{d_i}$ and contains the $d_i\times d_i$ identity matrix 
$I_{d_i}$ for all $i\in\Omega$. 
\par
The $\ast$-algebra of the subsystem with units in a subset $\nu\subseteq\Omega$ 
is the tensor product $\cA_\nu:=\bigotimes_{i\in \nu}\cA_i$. 
We omit the tensor product symbol $\otimes$ when no confusion arises. That is
to say, we write $AB$ in place of $A\otimes B$ for two matrices $A,B$. We denote 
the multiplicative identity of $\cA_\nu$ by $I_\nu$, and write  
$\bar\nu=\Omega\setminus\nu$. The \emph{partial trace} 
$\Tr_{\bar\nu}:\cA_\Omega\to\cA_\nu$ over the subsystem $\bar\nu$ is the adjoint 
to the embedding $\cA_\nu\to\cA_\Omega$, $A\mapsto A I_{\bar\nu}$. The matrix 
$\Tr_{\bar\nu}(\rho)$ is a state in $\cD(\cA_\nu)$, called the \emph{marginal} 
or \emph{reduced density matrix} in physics \cite{Erdahl1972,Zeng-etal2019},
for every state $\rho\in\cD(\cA_\Omega)$. Let
\begin{equation}\label{eq:marmap}\textstyle
\mar_{(\fg,\fa)}: \quad 
\cH(\cA_\Omega)\to\bigtimes_{\nu\in\fg}\cH(\cA_\nu), \qquad
A\mapsto (\Tr_{\bar\nu}(A))_{\nu\in\fg}
\end{equation}
denote the map which assigns marginals with respect to the pair $(\fg,\fa)$.
A \emph{$\fg$-local Hamiltonian} \cite{Chen-etal2012b} (also, 
quasi-local Hamiltonian \cite{Karuvade-etal2019}) is an element of the real 
vector space of hermitian matrices
\begin{equation}\label{eq:locHamilton}\textstyle
\cU(\fg,\fa)
=\left\{\sum_{\nu\in\fg}A_\nu I_{\bar\nu} : 
 A_\nu\in\cH(\cA_\nu),\nu\in\fg\right\}.
\end{equation}
We write $\cP_0(\fg,\fa)=\cP_0(\cU(\fg,\fa))\cup\{0\}$ to
denote the lattice of ground projectors of $\cU(\fg,\fa)$.
\par
In statistics \cite{Lauritzen1996}, the space $\,\cU(\fg,\fa)$ is known as a
\emph{hierarchical model subspace}. Strictly speaking, one has to distinguish
between the quantum mechanical concept of a Hamiltonian, or more generally of 
an observable, and its mathematical representation in terms of a hermitian 
matrix or a self-adjoint operator \cite{BengtssonZyczkowski2017}. 
As it is common in theoretical physics \cite{Zeng-etal2019}, we refer 
with a local Hamiltonian to a matrix. Similarly, we apply to
following notions of an interaction and of a frustration-free Hamiltonian to
matrices.
\par
Without changing the space $\cU(\fg,\fa)$, one can reduce $\fg$ to the antichain 
of its maximal elements (partially ordered by inclusion) and one can augment 
$\fg$ by adding all subsets of its elements as new elements. In the reduced 
form, $\fg$ is known as the \emph{generating class} of $\cU(\fg,\fa)$ in 
statistics \cite{Lauritzen1996,Geiger-etal2006}. If  $\fg$ has the augmented 
form we call $\fg$ a \emph{hypergraph}.
\par
It is useful to decompose local Hamiltonians into interaction terms. A matrix 
$A\in\cH(\cA_\Omega)$ is a \emph{$\nu$-factor interaction}, 
$\nu\subseteq\Omega$, if $A\in\cU(\{\nu\},\fa)$ and $A$ is perpendicular to 
$\cU(\{\mu\},\fa)$ for all $\mu\subset\nu$. In statistics, $\nu$-factor 
interactions are called \emph{$|\nu|$-factor interactions}
\cite[Section B.2]{Lauritzen1996}. If $\fg$ is a hypergraph, then $\cU(\fg,\fa)$ 
is the direct sum
\begin{equation}\label{eq:directUga}\textstyle
\cU(\fg,\fa)
=\bigoplus_{\nu\in\fg}\{
\mbox{$A\in\cH(\cA_\Omega)$ is a $\nu$-factor interaction}\}.
\end{equation}
We  construct a basis for each summand in the direct sum \eqref{eq:directUga}. 
Let $\cB_i$ be an orthogonal basis of $\cH(\cA_i)$, $i\in\Omega$. Then the matrices 
$B_1B_2\ldots B_N$, where $B_i\in\cB_i$, $i\in\Omega$, are an orthogonal basis 
of $\cH(\cA_\Omega)$. If $\cB_i$ contains the identity matrix $I_{d_i}$ for each 
$i\in\Omega$, then the set
\[
\left\{ B_1B_2\ldots B_N \mid
\mbox{$B_i\in\cB_i$ and $B_i=I_{d_i}$ if and only if $i\in\bar\nu$ 
for all $i\in\Omega$}\right\}
\]
is an orthogonal basis for the space of $\nu$-factor 
interactions, the dimension of which is therefore 
$\,\prod_{i\in\nu}(\dim\cH(\cA_i)-1 )$. 
\par
As per $\,\mar_{(\fg,\fa)}=\mar_{(\fg,\fa)}\circ\pi_{\cU(\fg,\fa)}$, the map 
$\,\mar_{(\fg,\fa)}\,$ factors through the space $\cU(\fg,\fa)$. The map 
restricts to the linear isomorphism 
\[
\cU(\fg,\fa)
\stackrel{\mar_{(\fg,\fa)}}{\longrightarrow}
\mar_{(\fg,\fa)}(\cH(\cA_\Omega)),
\]
as its injectivity follows from equation \eqref{eq:directUga}. The map restricts 
to the bijection
\begin{equation}\label{eq:mar-jnr}
\pi_{\cU(\fg,\fa)}(\cD(\cA_\Omega))
\stackrel{\mar_{(\fg,\fa)}}{\longrightarrow}
\mar_{(\fg,\fa)}(\cD(\cA_\Omega)),
\end{equation}
between the joint numerical range and the set of marginals. The dimension of the set 
of marginals is therefore $\,\dim(\cU(\fg,\fa))-1$. 
\par
In the sequel, we focus mainly on three-body systems, where $\Omega=\{1,2,3\}$. Up 
to permutations, there are only two generating classes with overlapping subsets,  
the edge sets $\{\{1,2\},\{2,3\}\}$ and $\{\{1,2\},\{2,3\},\{3,1\}\}$ of the 
\emph{path graph} $P_3$ and the \emph{cycle graph} $C_3$, respectively. We denote 
their hypergraphs by
\[
\fp_3=\Big\{\emptyset,\{1\},\{2\},\{3\},\{1,2\},\{2,3\}\Big\}
\]
and 
\[
\fc_3=\Big\{\emptyset,\{1\},\{2\},\{3\},\{1,2\},\{2,3\},\{3,1\}\Big\},
\]
respectively. By Equation~\eqref{eq:directUga}, we have 
\[\textstyle
\dim\cU(\fc_3,\fa)
= \prod_{i=1}^3\dim\cH(\cA_i) - \prod_{i=1}^3( \dim\cH(\cA_i) - 1 )
\]
and 
\[
\dim\cU(\fp_3,\fa)
= \dim\cU(\fc_3,\fa) - ( \dim\cH(\cA_1) - 1 )( \dim\cH(\cA_3) - 1 ).
\]
This gives $\,\dim\cU(\fc_3,\fa)=37$ and $\,\dim\cU(\fp_3,\fa)=28$ for 
three qubits.
\par
We denote the set of all subsets of cardinality $k$ of $\Omega$ by ${N\choose k}$. 
An ${N\choose k}$-local Hamiltonian is called a \emph{$k$-local Hamiltonian} 
\cite{Zeng-etal2019} and $\,\mar_{({N\choose k},\fa)}(\cD(\cA_\Omega))\,$ is the 
set of \emph{$k$-body marginals}. We will focus on the interaction pattern 
${3\choose 2}=\fc_3$ of three-body systems.
\par
A special class of local Hamiltonians appears in information theory and statistical 
mechanics frequently. A matrix $A\in\cH(\cA_\Omega)$ is a 
\emph{frustration-free Hamiltonian} 
\cite{Ji-etal2011,Movassagh-etal2010,Zeng-etal2019} with respect to 
the pair $(\fg,\fa)$ if there are $A_\nu\in\cH(\cA_\nu)$, $\nu\in\fg$, such that 
\[\textstyle
A=\sum_{\nu\in\fg}A_\nu I_{\bar\nu},
\]
and such that the ground projectors satisfy 
$P_0(A)\preceq P_0(A_\nu I_{\bar\nu})$ with respect to the Loewner order 
for all $\nu\in\fg$. Hence, the set of ground projectors of all frustration-free 
Hamiltonians, together with the zero projector, is the set 
\begin{equation}\label{eq:P0ff}\textstyle
\cP_0^\mathrm{ff}(\fg,\fa)
=\left\{ \bigwedge_{\nu\in\fg} P_\nu I_{\bar\nu} 
 : P_\nu\in\cP(\cA_\nu),\nu\in\fg \right\}.
\end{equation}
As per the associativity of the infimum, the infimum of any subset of 
$\cP_0^\mathrm{ff}(\fg,\fa)$ in the Loewner order on $\cP(\cA_\Omega)$ lies in 
$\cP_0^\mathrm{ff}(\fg,\fa)$. Hence, $\cP_0^\mathrm{ff}(\fg,\fa)$ is a complete 
lattice \cite[Section~I.4]{Birkhoff1973}. Furthermore, the lattice 
$\cP_0^\mathrm{ff}(\fg,\fa)$ is coatomistic. The set of coatoms of 
$\cP_0^\mathrm{ff}(\fg,\fa)$ is
\begin{equation}\label{eq:coatomsPff}\textstyle
\bigcup_{\nu\in\fg}
\left\{ P_\nu I_{\bar\nu} : \mbox{$P_\nu$ is a coatom of $\cP(\cA_\nu)$} \right\}
\end{equation}
if $\fg$ is a generating class. 
\par
%
%
%
\subsection{Probability Distributions that Factor}
\label{sec:factor}
A probability distribution factors if and only if it satisfies a set of 
polynomial equations and if its support set satisfies a certain condition 
\cite{Geiger-etal2006}. Here we show that the latter condition means that 
the support set is the ground projector of a frustration-free Hamiltonian.
\par
Let $\cA_i=\C^{C(d_i)}$ be the $\ast$-algebra of complex functions on the 
configuration space $C(d_i)=\{0,1,\dots,d-1\}$, $i\in\Omega$, introduced in
Section~\ref{subsec:diagonal-matrices}. The algebra $\cA_\nu$ of the subsystem 
$\nu\subseteq\Omega$ is the set $\cA_\nu=\C^{C_\nu}$ of complex functions on the 
configuration space 
\[
C_\nu=\bigtimes_{i\in\nu}C_i.
\]
If $x=(x_i)_{i\in\Omega}$ is an element of $C_\Omega$ and $\nu\subseteq\Omega$, 
then $x_\nu$ denotes the truncation of $x$ to $\nu$, that is to say, 
$x_\nu=((x_\nu)_i)_{i\in\nu}$ is the element of $C_\nu$ which satisfies 
$(x_\nu)_i=x_i$ for all $i\in\nu$. Let $\fg$ be a family of subsets of 
$\Omega$ and let
\[
C_\fg=\{(\nu,y)\mid y\in C_\nu,\nu\in\fg\}
\]
denote the disjoint union of the configuration spaces $C_\nu$, $\nu\in\fg$.
The matrix $M=(m_{(\nu,y),x})$ of the map \eqref{eq:marmap} has the 
coefficients
\begin{align}\label{eq:M}
m_{(\nu,y),x}
&=\mar_{(\fg,\fa)}(\delta_x)(\nu,y)\\
&=\Tr_{\bar\nu}(\delta_x)(y)
=\delta_{x_\nu}(y),
\qquad
(\nu,y)\in C_\fg, x\in C_\Omega,\nonumber
\end{align}
with respect to the bases $(\delta_x)_{x\in C_\Omega}$ of $\R^{C_\Omega}$ and 
$(\delta_{(\nu,y)})_{(\nu,y)\in C_\fg}$ of $\bigtimes_{\nu\in\fg}\R^{C_\nu}$.
The set of marginals is the convex hull of the columns of the matrix $M$, which
is called the marginal polytope \cite{Wang-etal2019}.
\par
By definition, a probability distribution $P\in\cD(\cA_\Omega)$ \emph{factors} with 
respect to $\fg$ if there exist a function $\psi_\nu:C_\nu\to\R$ for each $\nu\in\fg$ 
such that 
\[\textstyle
P(x)=\prod_{\nu\in\fg}\psi_\nu(x_\nu),
\qquad x\in C_\Omega.
\]
It is well known \cite{Geiger-etal2006} that a probability distribution 
$P\in\cD(\cA_\Omega)$ factors with respect to $\fg$ if and only if $P=P_\theta$ 
for some $\theta\in[-\infty,\infty)^{C_\fg}$, where
\begin{equation}\label{eq:log-lin-exp}\textstyle
P_\theta(x)=\frac{1}{Z(\theta)}e^{\scp{\theta,T(x)}},
\qquad x\in C_\Omega.
\end{equation}
Here, $T(x)=(m_{(\nu,y),x})_{(\nu,y)\in C_\fg}$ is the column with index 
$x\in C_\Omega$ of the matrix $M$ defined above in Equation~\eqref{eq:M}. The 
bracket $\scp{\cdot,\cdot}$ is the inner product on $\R^{C_\fg}$ restricted to 
nonnegative values in the second argument and extended to minus infinity in the 
first, by defining $(-\infty)\cdot 0=0$ and $(-\infty)\cdot t=-\infty$ for all 
$t>0$. The number $Z(\theta)$ is a normalization constant. As $e^{-\infty}=0$, 
the Equation~\eqref{eq:log-lin-exp} defines a probability distribution 
if and only if $\scp{\theta,T(x)}>-\infty$ holds for at least one $x\in C_\Omega$. 
Parametric models of the form \eqref{eq:log-lin-exp} are called 
\emph{hierarchical models} in the literature \cite{Lauritzen1996,Ay-etal2017},
they are special cases of \emph{exponential families} or \emph{log-linear models} 
\cite{Geiger-etal2006}.
\par
Probability distributions that factor have been characterized in terms of support sets
and commutative algebra. A subset $F\subseteq C_\Omega$ is \emph{$M$-feasible} if
\[\textstyle
\supp T(x)\not\subseteq\bigcup_{y\in F}\supp T(y),
\qquad\text{for all $x\in C_\Omega\setminus F$.}
\]
The \emph{nonnegative toric variety} $X_M$ is the set of all vectors 
$\xi=(\xi_x)_{x\in C_\Omega}$ in $[0,\infty)^{C_\Omega}$ which satisfy
\[\textstyle
\prod_{x\in C_\Omega}\xi_x^{u_x}
=\prod_{x\in C_\Omega}\xi_x^{v_x}
\]
whenever $u=(u_x)_{x\in C_\Omega},v=(v_x)_{x\in C_\Omega}\in\Z^{C_\Omega}$ are 
vectors of nonnegative integers such that $u-v$ is in the kernel of $M$. 
\par
\begin{Thm}[Geiger et al.~\cite{Geiger-etal2006}]\label{thm:geiger}
Let $\cA_i=\C^{C(d_i)}$ for all $i\in\Omega$ and let $P\in\cD(\cA_\Omega)$ be 
a probability distribution. Then $P$ factors with respect to $\fg$ if and only 
if the support of $P$ is $M$-feasible and $P$ lies in the nonnegative toric 
variety $X_M$.
\end{Thm}
We describe the support condition of Theorem~\ref{thm:geiger} in terms of ground 
projectors. We use the map \eqref{eq:pro-Cd} to identify projectors in $\cA_\nu$ 
and subsets of $C_\nu$, $\nu\subseteq\Omega$. The complementary projector 
to $P\subseteq C_\nu$ is $P'=C_\nu\setminus P=I_\nu-P$.
\par
\begin{Lem}\label{lem:factorff}
Let $\cA_i=\C^{C(d_i)}$ for all $i\in\Omega$ and let $P\in\cP(\cA_\Omega)$ be a 
projector. The following assertions are equivalent.
\begin{enumerate}
\item 
$P$ is $M$-feasible, 
\item
$P=\bigcap_{x\in P'}\bigcap_{\nu\in\fg_x}\{x_\nu\}' I_{\bar\nu}$,
where $\fg_x=\{\nu\in\fg\mid \text{$x_\nu\neq y_\nu$ for all $y\in P$}\}$,
\item
there exists a frustration-free $\fg$-local Hamiltonian 
$A\in\cU(\fg,\fa)$ such that $P=P_0(A)$ is the ground projector of $A$.
\end{enumerate}
\end{Lem}
\begin{proof}
Let $P$ be $M$-feasible. Then for all $x\in P'$ there exists $\nu\in\fg$
such that $x_\nu\neq y_\nu$ holds for all $y\in P$. This shows that the set 
$\fg_x$ above is nonempty for all $x\in P'$. Without any assumptions 
on the projector $P$, the inclusion ``$\subseteq$'' of the assertion 2) holds. 
Since each of the sets $\fg_x$, $x\in P'$, is nonempty, the right-hand side of
the equation 2) cannot contain any points of $P'$. This proves 
1) $\Rightarrow$ 2). The implication 2) $\Rightarrow$ 3) was discussed in 
Equation~\eqref{eq:P0ff}.
\par
Let $P=\bigcap_{\nu\in\fg} P_\nu I_{\bar\nu}$, where 
$P_\nu\in\cP(\cA_\nu)$ for all $\nu\in\fg$, and let $x\in P'$. Then there exists
$\nu\in\fg$ such that $x\not\in P_\nu I_{\bar\nu}$, that is to say,
$x_\nu\not\in P_\nu$. Since $y_\nu\in P_\nu$ holds for all $y\in P$, this proves
that $P$ is $M$-feasible, hence 3) $\Rightarrow$ 1).
\end{proof}
\begin{Cor}\label{cor:factor}
Let $\cA_i=\C^{C(d_i)}$, $i\in\Omega$, and let $P\in\cP(\cA_\Omega)$ be a nonzero 
projector. Then $P$ is the ground projector of a frustration-free $\fg$-local 
Hamiltonian if and only if there are functions $A_\nu\in\R^{C_\nu}$, $\nu\in\fg$, 
such that $P=\prod_{\nu\in\fg}A_\nu I_{\bar\nu}$. If this is the case, then there 
are projectors $P_\nu\in\cP(\cA_\nu)$, $\nu\in\fg$, such that 
$P=\prod_{\nu\in\fg}P_\nu I_{\bar\nu}$.
\end{Cor}
\begin{proof}
If the projector $P$ factors, then the uniform probability distribution $\frac{P}{|P|}$ 
on the set $P$ factors. Theorem~\ref{thm:geiger} then shows that $P$ is $M$-feasible 
and Lemma~\ref{lem:factorff} concludes that $P$ is the ground projector of a 
frustration-free Hamiltonian. Conversely, if $P$ is the ground projector of 
a frustration-free Hamiltonian, then Equation~\eqref{eq:P0ff} shows  
$P=\bigcap_{\nu\in\fg} P_\nu I_{\bar\nu}$, where $P_\nu\in\cP(\cA_\nu)$,
$\nu\in\fg$. This proves the claim, as 
$\bigcap_{\nu\in\fg} P_\nu I_{\bar\nu}
=\prod_{\nu\in\fg} P_\nu I_{\bar\nu}$.
\end{proof}
%
%
\subsection{A First Glimpse at Three Qubits}
\label{sec:first-glimpse}
We consider a system of $N\in\N$ qubits. The algebras 
\[\textstyle
\fa_{N,\mathrm{qu}}=(\underbrace{M_2,M_2,\ldots,M_2}_{\text{$N$ copies}})
\] 
of the units are all equal to the algebra $M_2$ of a single qubit 
(Section~\ref{sec:qubitalgebra}). The algebra
\[\textstyle
\cA_\Omega=M_2^{\otimes N}=M_{2^N}
\]
of the full system is the $N$-fold tensor product of $M_2$.
The space of hermitian matrices $\cH(\cA_\Omega)$ has the orthogonal basis 
\[\textstyle
\{A_1A_2\ldots A_N : A_i\in \{I,X,Y,Z\}, i\in\Omega=\{1,2,\ldots,N\} \}.
\] 
\par
We begin with an observation regarding the space $\cU(\fc_3,\qu)$ of two-local 
three-qubit Hamiltonians.
\par
\begin{Lem}\label{lem:no-seven}
The space $\,\cU(\fc_3,\qu)\,$ contains no matrix of rank one. In other words, 
the lattice of ground projectors $\cP_0(\fc_3,\qu)$ contains no projector 
of rank seven.
\end{Lem}
\begin{proof}
Let $P=\ket{\psi}\!\!\bra{\psi}$ be the projector onto the line spanned by a 
unit vector $\psi\in(\C^{C(2)})^{\otimes 3}$. It is known \cite{Acin-etal2001} that, 
up to a local unitary transformation, there are 
$\kappa_0,\kappa_1,\ldots,\kappa_4\geq 0$ and $\theta\in[0,\pi)$ such that 
\[
\psi
=\kappa_0 e^{\ii\theta}\ket{000}
+\kappa_1\ket{001}
+\kappa_2\ket{010}
+\kappa_3\ket{100}
+\kappa_4\ket{111}.
\]
Let us assume that $P$ lies in $\cU(\fc_3,\qu)$. Then the inner products of 
$P$ with all three-factor interactions vanish. In particular 
\[
0=\scp{P,ZZZ}=\kappa_0^2-\kappa_1^2-\cdots-\kappa_4^2,
\]
which is only possible if $\kappa_0=1/\sqrt{2}$, as $\psi$ is a unit vector. 
Hence,
\begin{align*}
\scp{P,ZZX} &= \sqrt{2}\,\kappa_1\cos(\theta), &
\scp{P,ZXZ} &= \sqrt{2}\,\kappa_2\cos(\theta), \\
\scp{P,XZZ} &= \sqrt{2}\,\kappa_3\cos(\theta), &
\scp{P,XXX} &= \sqrt{2}\,\kappa_4\cos(\theta),
\end{align*}
which shows $\kappa_1=\kappa_2=\kappa_3=\kappa_4=0$ if $\theta\neq\frac\pi{2}$ 
modulo $\pi$. Also,
\begin{align*}
\scp{P,ZZY} &= -\sqrt{2}\,\kappa_1\sin(\theta), &
\scp{P,ZYZ} &= -\sqrt{2}\,\kappa_2\sin(\theta), \\
\scp{P,YZZ} &= -\sqrt{2}\,\kappa_3\sin(\theta), &
\scp{P,YYY} &= \sqrt{2}\,\kappa_4\sin(\theta),
\end{align*}
shows $\kappa_1=\kappa_2=\kappa_3=\kappa_4=0$ if $\theta=\frac\pi{2}$ modulo $\pi$. 
In any case,
\[
1/2=\kappa_0^2+\kappa_1^2+\kappa_2^2+\kappa_3^2+\kappa_4^2
=\braket{\psi|\psi}=1
\]
is a contradiction.
\end{proof}
%
%
%
\subsection{All About Three Bits}
\label{sec:all3bits}
We consider a system of $N\in\N$ bits, which is a special case of the
setup discussed in Section~\ref{sec:factor}. The configuration space of 
a bit is $C(2)=\{0,1\}$. The algebras 
\[
\fa_{N,\mathrm{cl}}
=(\underbrace{\C^{C(2)},\C^{C(2)},\ldots,\C^{C(2)}}_{\text{$N$ copies}})
\]
of the units are all equal to the algebra $\C^{C(2)}$ of $2$-by-$2$ diagonal 
matrices, associated with a single bit (Section~\ref{sec:bitalgebra}). The 
algebra 
\[
\cA_\Omega=(\C^{C(2)})^{\otimes N}=\C^{C_\Omega}
\]
of the full system is the $N$-fold tensor product of $\C^{C(2)}$, which is the 
set of complex functions on $C_\Omega$. The set 
\[
\{A_1A_2\ldots A_N : A_i\in \{I,Z\}, i\in\Omega  \}
\] 
is an orthogonal basis of the space of hermitian matrices 
$\cH(\cA_\Omega)=\R^{C_\Omega}$. We identify two representations of rank-one
projectors in $\cA_\Omega$, using the isomorphism of Equation~\eqref{eq:pro-Cd},
\begin{equation}\label{eq:label-rank1}
x_1x_2\ldots x_N
=
\tfrac{1}{2^N}\big(I+(-1)^{x_1}Z\big)\big(I+(-1)^{x_2}Z\big)
\ldots\big(I+(-1)^{x_N}Z\big),
\end{equation}
for all $N$-digit binary numbers $x=x_1x_2\ldots x_N\in C_\Omega$. On the left-hand 
side of Equation~\eqref{eq:label-rank1} there is an element of the configuration 
space $C_\Omega$, and on the right-hand side there is a diagonal $2^N\times 2^N$ 
matrix. The number $x$ marks the position of the diagonal entry $1$ of this matrix, 
which has all other entries equal to $0$. The position increases from $x=00\dots0$ 
at the top left to $x=11\dots1$ at the bottom right of the diagonal.
\par
We focus on $N=3$ where $\Omega=\{1,2,3\}$. The coatoms of the lattice of ground 
projectors of the space $\,\cU(\fc_3,\cl)\,$ of two-local three-bit Hamiltonians are 
in a one-to-one correspondence with the edges of the graph $K_{4,4}$. 
We simplify the proof  \cite{Weis2018b} of this statement in Lemma~\ref{lem:3BitG32} 
below. We also describe the ground projectors of frustration-free Hamiltonians, in 
Lemma~\ref{lem:3BitG32ff}, and of Hamiltonians interacting along a path without a 
cycle, in Lemma~\ref{lem:3BitGpath}.
\par
Note that the space $\cU(\fc_3,\cl)$ is the orthogonal complement to $f=ZZZ$ in 
$\cH(\cA_\Omega)$. We have 
\[
f(x_1x_2x_3)=(-1)^{x_1+x_2+x_3}, \qquad x_1,x_2,x_3\in C(2),
\]
\[
f=\diag(+1,-1,-1,+1,-1,+1,+1,-1).
\]
We identify the vertex set of the complete bi-partite graph $K_{4,4}$ with 
$C_\Omega$, the bi-partition being defined by the two fibers of $f$. In other 
words, $\{x,y\}\subseteq C_\Omega$ is an edge of $K_{4,4}$ if and only if the 
digit sums of $x$ and $y$ differ modulo two.
\par
The projectors in $\cA_\Omega$ are in a one-to-one correspondence with the subsets 
of $C_\Omega$. The complementary projector to $P\subseteq C_\Omega$ is 
$P'=C_\Omega\setminus P=III-P$.
\par
\begin{Lem}\label{lem:3BitG32}
Let $P\subseteq C(2)^{\times 3}$. The projector $P$ is a coatom of 
$\cP_0(\fc_3,\cl)$ if and only if $P'$ is an edge of the graph $K_{4,4}$. 
The projector $P$ lies in $\cP_0(\fc_3,\cl)$ if and only if $P'$ is a union 
of edges of $K_{4,4}$ (possibly empty).
\end{Lem}
\begin{proof}
We abbreviate $\cU=\cU(\fc_3,\cl)$ and $\cP_0=\cP_0(\fc_3,\cl)$. The lattice $\cP_0$ 
has no elements of rank seven by Lemma~\ref{lem:no-seven}. 
\par
If $P\in\cP_0$ has rank at most five, then $P$ is not a coatom. Indeed, as $f$ is 
perpendicular to $\cU$, Lemma~\ref{lem:Kp=0} shows that $f$ is nonconstant on $P'$. 
Hence, there are mutually distinct points $x,y,z\in P'$ such that 
$f(x)\neq f(y)=f(z)$. Then the space $\cH(P'\!.\cA_\Omega.P')\cap\cU$ has dimension 
at least two, as it contains the linearly independent rank-two projectors $\{x,y\}$ 
and $\{x,z\}$. According to Equation \eqref{eq:experimental-app}, this proves that 
$P$ is not a coatom of $\cP_0$.
\par
The preceding part of the proof shows that a projector $P\subseteq C(2)^{\times 3}$
is a coatom of $\cP_0$ if and only if $P\in\cP_0$ and $|P|=6$. Let $P'=\{x,y\}$ with
$x\neq y$ and consider the cone $\cK(P)=P'\!.\cA_\Omega^+.P'\cap\cU$ defined in 
Equation~\eqref{eq:KP}. If $f(x)\neq f(y)$ then $\cK(P)$ is the ray spanned by 
$\{x,y\}$. If $f(x)=f(y)$ then $\cK(P)=\{0\}$. Thus, Lemma~\ref{lem:charPU} completes 
the assertion on coatoms.
\par
The second assertion is true since the infimum in $\cP_0$ is the intersection and 
since $\cP_0$ is coatomistic by Lemma~\ref{lem:ray-coatoms}.
\end{proof}
\begin{table}[t]
\begin{align*}
\{000,001\}
 &= \diag(1,1,0,0,0,0,0,0)
 = \tfrac{1}{4}(I+Z)(I+Z)I\\
\{010,011\}
 &= \diag(0,0,1,1,0,0,0,0)
 = \tfrac{1}{4}(I+Z)(I-Z)I\\
\{100,101\}
 &= \diag(0,0,0,0,1,1,0,0)
 = \tfrac{1}{4}(I-Z)(I+Z)I\\
\{110,111\}
 &= \diag(0,0,0,0,0,0,1,1)
 = \tfrac{1}{4}(I-Z)(I-Z)I\\
\{000,010\}
 &= \diag(1,0,1,0,0,0,0,0)
 = \tfrac{1}{4}(I+Z)I(I+Z)\\
\{001,011\}
 &= \diag(0,1,0,1,0,0,0,0)
 = \tfrac{1}{4}(I+Z)I(I-Z)\\
\{100,110\}
 &= \diag(0,0,0,0,1,0,1,0)
 = \tfrac{1}{4}(I-Z)I(I+Z)\\
\{101,111\}
 &= \diag(0,0,0,0,0,1,0,1)
 = \tfrac{1}{4}(I-Z)I(I-Z)\\
\{000,100\}
 &= \diag(1,0,0,0,1,0,0,0)
 = \tfrac{1}{4}I(I+Z)(I+Z)\\
\{001,101\}
 &= \diag(0,1,0,0,0,1,0,0)
 = \tfrac{1}{4}I(I+Z)(I-Z)\\
\{010,110\}
 &= \diag(0,0,1,0,0,0,1,0)
 = \tfrac{1}{4}I(I-Z)(I+Z)\\
\{011,111\}
 &= \diag(0,0,0,1,0,0,0,1)
 = \tfrac{1}{4}I(I-Z)(I-Z)
\end{align*}
\caption{\label{tab:12exposed}%
Edges of $K_{4,4}$ that connect vertices differing in exactly one digit.}
\end{table}
\begin{table}[t]
\begin{align*}
\{000,111\}
 &= \diag(1,0,0,0,0,0,0,1)
 = \tfrac{1}{4}(III+IZZ+ZIZ+ZZI)\\
\{001,110\}
 &= \diag(0,1,0,0,0,0,1,0)
 = \tfrac{1}{4}(III-IZZ-ZIZ+ZZI)\\
\{010,101\}
 &= \diag(0,0,1,0,0,1,0,0)
 = \tfrac{1}{4}(III-IZZ+ZIZ-ZZI)\\
\{100,011\}
 &= \diag(0,0,0,1,1,0,0,0)
 = \tfrac{1}{4}(III+IZZ-ZIZ-ZZI)
\end{align*}
\caption{\label{tab:4exposed}%
Edges of $K_{4,4}$ that connect vertices differing in all three digits.}
\end{table}
We describe the coatoms of $\cP_0(\fc_3,\cl)$ in terms of matrices and 
extreme points.
\par
\begin{Rem}[Edges, Matrices, and Extreme Points]\label{rem:gp-expo}
Lemma~\ref{lem:3BitG32} above characterizes the coatoms of $\cP_0(\fc_3,\cl)$ as 
those projectors $P\subseteq C(2)^{\times 3}$ for which the complementary 
projectors $P'$ are edges of the complete bi-partite graph $K_{4,4}$. 
Table~\ref{tab:12exposed} and Table~\ref{tab:4exposed} list the sixteen edges 
of $K_{4,4}$ in the matrix notation of Equation~\eqref{eq:label-rank1}. 
By Theorem~\ref{thm:coatoms-exposed}, the matrix $4P'-III$ is an extreme point of 
the spectrahedron $\cS(\cU(\fc_3,\cl))$ for all sixteen coatoms $P$ of 
$\cP_0(\fc_3,\cl)$, because they are two-local Hamiltonians. All extreme points 
of the spectrahedron $\cS(\cU(\fc_3,\cl))$ are obtained in this way.
\end{Rem}
\begin{Lem}\label{lem:3BitG32ff}
Let $P\subseteq C(2)^{\times 3}$. The projector $P$ is a coatom of the
lattice $\cP_0^\mathrm{ff}(\fc_3,\cl)$ if and only if $P'$ is an edge of the graph 
$K_{4,4}$ which connects two vertices that differ in exactly one digit. The 
projector $P$ lies in $\cP_0^\mathrm{ff}(\fc_3,\cl)$ if and only if $P'$ is a 
union of the described edges. 
\end{Lem}
\begin{proof}
Let $\cP_0^\mathrm{ff}=\cP_0^\mathrm{ff}(\fc_3,\cl)$. By Equation 
\eqref{eq:coatomsPff}, the set of coatoms of $\cP_0^\mathrm{ff}$ is 
\[\textstyle
\bigcup_{i=1}^3
\left\{ P I_{\{i\}} : 
 \mbox{$P$ is a coatom of $\cP(\cA_{\Omega\setminus\{i\}})$} \right\}.
\]
The complementary projector to the coatom $P I_{\{i\}}$ can be written as 
\[
(P I_{\{i\}})'
= I_\Omega-P I_{\{i\}}
= Q I_{\{i\}},
\]
where $Q=I_{\Omega\setminus\{i\}}-P$ is an atom of the lattice 
$\cP(\cA_{\Omega\setminus\{i\}})$. Since $\cA_{\Omega\setminus\{i\}}$ is isomorphic 
to the two-bit algebra $\C^{C(2)}\otimes\C^{C(2)}$, the projector $Q$ is a rank-one projector, which we write as a two-digit binar number $Q=xy$ for some $x,y\in C(2)$. 
This shows that the two elements in the subset 
\[
(P I_{\{i\}})'=Q I_{\{i\}}\subseteq C(2)^{\times 3}
\]
differ exactly in the $i$-th digit. Conversely, 
$(I_{\Omega\setminus\{i\}}\setminus\{xy\}) I_{\{i\}}$ is a coatom of 
$\cP_0^\mathrm{ff}$, again by Equation \eqref{eq:coatomsPff}, for all $i\in\Omega$ 
and $x,y\in C(2)$. The second statement is true as the infimum in $\cP_0^\mathrm{ff}$ 
is the intersection and since $\cP_0^\mathrm{ff}$ is coatomistic, see the discussion
in Section~\ref{sec:hierarchical}.
\end{proof}
We turn to the interaction pattern $\fp_3$ with generating class 
$\{\{1,2\},\{2,3\}\}$, the edge set of the path graph $P_3$. The space 
$\cU(\fp_3,\cl)$ is the orthogonal complement of the span of $f=ZZZ$ and $g=ZIZ$ in 
$\cH(\cA_\Omega)=\R^{C_\Omega}$. We have 
\[
g(x_1x_2x_3)=(-1)^{x_1+x_3}, \qquad x_1,x_2,x_3\in\{0,1\},
\]
\[
g=\diag(+1,-1,+1,-1,-1,+1,-1,+1).
\] 
\par
\begin{Lem}\label{lem:3BitGpath}
Let $P\subseteq C(2)^{\times 3}$. The projector $P$ is a coatom of the lattice 
$\cP_0(\fp_3,\cl)$ if and only if $P'$ is an edge of the graph $K_{4,4}$ which 
connects two vertices that differ exactly in the first digit or exactly in the 
third. The projector $P$ lies in $\cP_0(\fp_3,\cl)$ if and only if $P'$ is a 
union of the described edges. Every nonzero element of $\cP_0(\fp_3,\cl)$ is
the ground projector of a frustration-free $\fp_3$-local three-bit Hamiltonian.
\end{Lem}
\begin{proof}
We abbreviate $\cU=\cU(\fp_3,\cl)$ and $\cP_0=\cP_0(\fp_3,\cl)$. The lattice 
$\cP_0$ has no elements of rank seven by Lemma~\ref{lem:no-seven}. 
\par
If $P\in\cP_0$ has rank at most five, then $P$ is not a coatom. Indeed, since 
$f$ and $g$ are perpendicular to $\cU$, Lemma~\ref{lem:Kp=0} shows 
that both $f$ and $g$ are nonconstant on $P'$. As $|P'|\geq 3$, there are three 
mutually distinct points $x,y,x\in P'$, such that both $f$ and $g$ are 
nonconstant on $\{x,y,z\}$. First, let both $f$ and $g$ be nonconstant on a 
subset of size two of $\{x,y,z\}$. Then $P$ is not a coatom of $\cP_0$ by
a similar reasoning as in Lemma~\ref{lem:3BitG32} above. Otherwise, by 
multiplying $f$ and $g$ with $\pm1$ and permuting the labels of the points 
$x,y,z$, we can assume without loss of generality that
\[\begin{array}{rcccccccl}
+1 &=& f(x) &\neq & f(y) &=    & f(z) &=& -1,\\
+1 &=& g(x) &=    & g(y) &\neq & g(z) &=& -1.
\end{array}\]
Second, if $P'=\{x,y,z\}$ then the cones $\cK(P)$ and $\cK(\{x,z\}')$, defined 
in Equation~\eqref{eq:KP}, are equal to the ray spanned by $\{x,z\}$. 
Lemma~\ref{lem:charPU} then shows $P\not\in\cP_0$. Third, let 
$\eta\in P'\setminus\{x,y,z\}$. If $f(\eta)=f(a)$ and $g(\eta)=g(a)$ for a point
$a\in\{x,y,z\}$, then $P\not\in\cP_0$ again by a similar reasoning as in 
Lemma~\ref{lem:3BitG32}. Finally, if $f(\eta)=+1$ and $g(\eta)=-1$ then 
the space $\cH(P'\!.\cA_\Omega.P')\cap\cU$ contains the linearly independent
rank-two projectors $\{x,z\}$ and $\{y,\eta\}$, hence $P$ is not a coatom by 
Equation~\eqref{eq:experimental-app}. 
\par
Let $P\subseteq C(2)^{\times 3}$. The preceding part of the proof shows that $P$ 
is a coatom of $\cP_0$ if and only if $P\in\cP_0$ and $|P|=6$. Let $P'=\{x,y\}$ 
for $x\neq y$. As cone $\cK(p)$ is a ray if and only if $f(x)\neq f(y)$ and 
$g(x)\neq g(y)$, Lemma~\ref{lem:ray-coatoms} confirms the assertion on 
coatoms.
\par
The second assertion, regarding general elements of $\cP_0$, is true 
since the infimum in $\cP_0$ is the intersection and because $\cP_0$ is 
coatomistic (see the last paragraph of Section~\ref{sec:jnrs}).
\par
The third assertion follows from the fact that every coatom $P$ of $\cP_0$ is the 
ground projector of its complementary projector $P'$, and that $P'$ is a 
frustration-free $\fp_3$-local three-bit Hamiltonian, see the rows 1--4 or 9--12 
of Table~\ref{tab:12exposed}.
\end{proof}
Lemma~\ref{lem:3BitGpath}, Lemma~\ref{lem:3BitG32ff}, and Lemma~\ref{lem:3BitG32}
describe the lattices of ground projectors
\[
\cP_0^\mathrm{ff}(\fp_3,\cl)
=\cP_0(\fp_3,\cl)
\subset \cP_0^\mathrm{ff}(\fc_3,\cl)
\subset \cP_0(\fc_3,\cl).
\]
The coatoms are projectors of rank six which are complements to certain edges of 
the graph $K_{4,4}$. Eight edges belong to the frustration-free $\fp_3$-local 
Hamiltonians as well as to all $\fp_3$-local Hamiltonians 
(Table~\ref{tab:12exposed}, rows 1--4 and 9--12). Twelve edges pertain to the 
frustration-free two-local Hamiltonians (Table~\ref{tab:12exposed}), and sixteen 
edges to all two-local Hamiltonians (Table~\ref{tab:12exposed} and 
Table~\ref{tab:4exposed} together). All other lattice elements are intersections 
of coatoms.
\par
%
%
\subsection{A Family of Coatoms of Rank Five}
\label{sec:rank-three}
We present a family of coatoms of rank five in the lattice $\cP_0(\fc_3,\qu)$ 
of ground projectors of the space $\cU(\fc_3,\qu)$ of two-local three-qubit 
Hamiltonians. This is in contrast with the classical lattice $\,\cP_0(\fc_3,\cl)$, 
where Lemma~\ref{lem:3BitG32} rules out the existence of coatoms of rank five.
We discovered the family of coatoms with the help of the semidefinite programming 
strategy that samples extreme points from the dual spectrahedron proposed in 
Section \ref{sec:spectrahedra}. Whereas this is a two-parameter family, it covers 
a higher-dimensional family of extreme points in the spectrahedron of dimension up
to eleven. In fact \cite{Carteret-etal1999}, the generic dimension of the orbit of 
a mixed (or pure) $N$-qubit state under local unitary transformations is $3N$ if
$N\geq 2$. The question as to whether our family provides two nonlocal parameters
can be rigorously studied using invariant theory, see 
\cite{Grassl-etal1998,Rains2000,Sun-etal2017} and the references therein.
\par
%
%
\subsubsection{The Numerical Procedure}
Recall that we are trying to find extreme points of the spectrahedron 
$\cS(\cU(\fc_3,\qu))$. By what was seen in Section~\ref{sec:hierarchical}, this 
spectrahedron is given by
\[
\cS(\cU(\fc_3,\qu))=\left\{ x \in \mathbb{R}^{36} \, 
\colon \, I_8 + \sum_{i=1}^{36} x_i A_i \succeq 0 \right\},
\]
where the $A_i$ range over all the matrices of the form $B_1 B_2 B_3$ where 
$B_j\in\{I,X,Y,Z\}$, for $j=1,2,3$, at least one of them is $I$, 
but not all three are $I$.
\par
This is then a $36$-dimensional object defined by an $8 \times 8$ positive 
semidefinite condition, an object that is quite amenable to semidefinite 
programming. Using \texttt{MOSEK 9.2.10}, we optimized in randomly generated 
directions in $\mathbb{R}^{36}$ and recorded the ranks of the corresponding 
matrices. After $65000$ we recorded the following rank distribution
\begin{center}
\begin{tabular}{r|ccc}
Rank & $2$ & $3$ & $4$ \\ \hline
Frequency & $83.62 \%$ &$9.57 \%$ & $6.81 \%$\\
\end{tabular}.
\end{center}
Note that the ranks indicated are numerical, obtained by cutting off eigenvalues 
of sufficiently small magnitude, and do not provide exact certificates of the 
existence of such extreme points. This, however, strongly suggests that in 
addition to the rank $6$ coatoms in $\cU(\fc_3,\cl)$, there exist rank $4$ and 
$5$ coatoms in $\cU(\fc_3,\qu)$.
\par
By carefully looking at the samples we were obtaining with rank $3$, and after 
some ad hoc algebraic manipulations we were able to identify some of them that 
seem to come from the two-parameter family 
\begin{align*}
 M(a,t) &=\, {\tiny\left(
\begin{array}{cccccccc}
 a^2 & 0 & 0 & 0 & 0 & 0 & 0 & 0 \\
 0 & 0 & 0 & 0 & 0 & 0 & 0 & 0 \\
 0 & 0 & 4\cos(t)^2 & -\sqrt{\eta}\sin(2t) & 0 & 0 & 0 & 0 \\
 0 & 0 & -\sqrt{\eta}\sin(2t) & \eta\sin(t)^2 & 0 & 0 & 0 & 0 \\
 0 & 0 & 0 & 0 & 4\sin(t)^2 & \sqrt{\eta}\sin(2t) & 0 & 0 \\
 0 & 0 & 0 & 0 & \sqrt{\eta}\sin(2t) & \eta\cos(t)^2 & 0 & 0 \\
 0 & 0 & 0 & 0 & 0 & 0 & 0 & 0 \\
 0 & 0 & 0 & 0 & 0 & 0 & 0 & 0 \\
\end{array}
\right)}\\[.5\baselineskip]
 &=\, III 
 \,+\, \frac{a^2}{4} \,[\, IIZ \,+\, \sin(t)^2IZI \,+\, \cos(t)^2\,ZII \,]\\
 &+\, \frac{\sqrt{\eta}}{2}\sin(2t)
 \,(\, IZ - ZI \,)\,X\\
 &+\,\frac{a^2}{8}\,(\,IZ+ZI\,)\,Z
 \,-\,\frac{1}{8}(8-a^2)\cos(2t)\,(\,IZ-ZI\,)\,Z 
 \,-\, \frac{\eta}{4}\, ZZI,
\end{align*}
where $\eta=4-a^2$, $0\leq a\leq 2$, and $t\in[0,\pi)$.
\par
This suggests that the matrices in this family correspond to a family of extreme 
points of $\cS(\cU(\fc_3,\qu))$ whose normal cones collectively have some 
non-negligible positive volume. The last task remaining is to, from this 
heuristically derived family, derive an exact certificate that it is indeed a 
family of extreme points. In what follows, we illustrate that procedure.
\par
%
%
\subsubsection{Algebraic Certificates}
Excluding special parameters, we show that the ground projector $P=P_0(M(a,t))$ 
is a coatom of the lattice $\cP_0(\fc_3,\qu)$. By 
Equation~\eqref{eq:experimental-app} in Remark~\ref{ref:experi}, it suffices to
show that the intersection of the space $\cH(P'\!.\cA_\Omega.P')$ of hermitian 
matrices in the algebra  $P'\!.\cA_\Omega.P'$ with the space $\cU(\fc_3,\qu)$ 
of two-local three-qubit Hamiltonians is the line spanned by $M(a,t)$. 
If $P$ is a coatom, then we also learn from Theorem~\ref{thm:coatoms-exposed} 
that the hermitian matrix $M(a,t)-III$ is an extreme point of the spectrahedron 
$\cS(\cU(\fc_3,\qu))$.
Here, $P'=III-P$, and $\cA_\Omega=\cA_{\{1,2,3\}}=M_2^{\otimes 3}$ is the algebra 
associated with three qubits. 
\par
Let us recover the matrix $M(a,t)$ from $P$ in the sense that 
$\cH(P'\!.\cA_\Omega.P')\cap\cU(\fc_3,\qu)$ is the line spanned by $M(a,t)$,
under the condition that $a\not\in\{0,2\}$ and $t\not\in\{0,\frac{\pi}{2}\}$.
The matrix $M(a,t)$ is positive semidefinite of rank three, and its kernel is 
spanned by the vectors
\begin{gather*}
\ket{001}, \qquad \ket{110}, \qquad \ket{111},\\
\ket{\psi_1}=\sqrt{\eta}\sin(2t) \ket{010} + 4\cos(t)^2 \ket{011},\\
\ket{\psi_2}=\sqrt{\eta}\sin(2t) \ket{100} - 4\sin(t)^2 \ket{101}.
\end{gather*}
Let $A\in\cU(\fc_3,\qu)$ be an arbitrary two-local Hamiltonian. We can write 
$A=\sum_{i=1}^{37} z_i A_i$ for some $z\in\R^{37}$, where the $A_i$ range over 
all the matrices of the form $B_1 B_2 B_3$ where $B_j\in\{I,X,Y,Z\}$, for 
$j=1,2,3$ and at least one of them is $I$. Assuming that $A$ lies in the algebra 
$P'\!.\cA_\Omega.P'$, we can set the real and imaginary parts of the vectors 
$A\ket{001}$, $A\ket{110}$, and $A\ket{111}$ to zero, as the vectors $\ket{001}$, 
$\ket{110}$, $\ket{111}$ lie in the kernel of every matrix from 
$P'\!.\cA_\Omega.P'$. This allows us to get rid of $29$ of the $z_i$'s with the 
help of \texttt{Wolfram Mathematica 9}. One eliminates by hand seven of the  
remaining eight parameters by requiring that the real and imaginary 
parts of the vectors $A\ket{\psi_1}$ and $A\ket{\psi_2}$ are zero. This is the 
only place where the variables $a$ and $t$ play a role, as far as 
$a\not\in\{0,2\}$ and $t\not\in\{0,\frac{\pi}{2}\}$ guarantees that the seven
variables can be eliminated. 
\par
Out of curiosity, we discuss the special parameter values. As the matrix $M(0,t)$ 
has rank two, Theorem~\ref{thm:coatoms-exposed} and Lemma~\ref{lem:no-seven} show 
that the point $M(0,t)-III$ is an extreme point of the spectrahedron 
$\cS(\cU(\fc_3,\qu))$. The corresponding coatom of the lattice $\cP_0(\fc_3,\qu)$
is the matrix $III-\frac{1}{4}M(0,t)$, as the positive eigenvalues of $M(0,t)$ are 
equal for all $t\in[0,\pi)$. The matrix $M(2,t)$ belongs to the commutative algebra 
$(\C^{C(2)})^{\otimes 3}$ and has rank three unless $t=0$ or 
$t=\frac{\pi}{2}$, in which case it has rank two. By Lemma~\ref{lem:3BitG32}, the 
ground projector $P_0(M(2,t))$ is not a coatom of $\cU(\fc_3,\cl)$ if $M(2,t)$ has 
rank three. It follows that the point $M(2,t)-III$ is an extreme point of the 
spectrahedron $\cS(\cU(\fc_3,\qu))$ if and only if $t\in\{0,\frac{\pi}{2}\}$. The 
matrix $M(a,0)$ belongs to the commutative algebra 
$(\C^{C(2)})^{\otimes 3}$ and has rank three unless $a=0$ or $a=2$, in 
which case it has rank two. It follows that $M(a,0)-III$ is an extreme point of 
the spectrahedron $\cS(\cU(\fc_3,\qu))$ if and only if $a\in\{0,2\}$. The same 
happens at $t=\frac{\pi}{2}$, where the point $M(a,\frac{\pi}{2})-III$ is an extreme 
point of the spectrahedron $\cS(\cU(\fc_3,\qu))$ if and only if $a\in\{0,2\}$.
\par
%
%
\subsection{Tomography and Nonexposed Faces}
\label{sec:nonexposed}
Karuvade et al.~\cite{Karuvade-etal2019} recently discovered a six-qubit state 
that is uniquely determined by its two-body marginals, but which is not the 
unique ground state of any two-local Hamiltonian. We discuss the convex geometric 
consequences of this result. Related observations have been made before
\cite{Chen-etal2012b,Rosina2000}.
\par
A basic problem of quantum state tomography is to find conditions under which a 
state can be recovered from certain data, for example from its image under the 
projection $\pi_\cU$ onto a space of hermitian matrices $\,\cU$. We say a state 
$\rho\in\cD(\cA)$ is \emph{uniquely determined by $\pi_\cU$} if $\rho=\sigma$ 
whenever $\pi_\cU(\rho)=\pi_\cU(\sigma)$ for all states $\sigma\in\cD(\cA)$. 
We say a subset of $\cD(\cA)$ is \emph{uniquely determined by $\pi_\cU$} if all 
its elements are. 
\par
Another problem of tomography is concerned with ground states. A state 
$\rho\in\cD(\cA)$ is a \emph{ground state} of a hermitian matrix $A\in\cH(\cA)$ 
if $\rho$ is supported by the ground projector $P_0(A)$, that is to say, if 
$\rho$ lies in $\phi_\cA(P_0(A))$, as defined in Equation \eqref{eq:iso-P-FD}. 
A state $\rho$ is the \emph{unique ground state} of $A$ if we have
$\{\rho\}=\phi_\cA(P_0(A))$. In this case $\rho=P_0(A)/\Tr(P_0(A))$ holds.
\par
The above notions of tomography have counterparts in terms of faces of the joint
numerical range $\pi_\cU(\cD(\cA))$. A \emph{face} of a convex set $C$ in a 
Euclidean space is a convex subset $F$ of $C$ such that whenever 
$(1-\lambda)x+\lambda y$ lies in $F$ for some $\lambda\in(0,1)$ and $x,y\in C$, 
then $x$ and $y$ are also in $F$. It is well known that every exposed face of 
$C$ is a face of $C$. A face that is not an exposed face is called a 
\emph{nonexposed} face. If $x\in C$ and $\{x\}$ is a face or nonexposed face, 
then $x$ is called an \emph{extreme point} or \emph{nonexposed point}, respectively.
\par
A subset $F$ of $\cD(\cA)$ is \emph{lift-invariant} under $\pi_\cU$ if 
$F=\pi_\cU|_{\cD(\cA)}^{-1}(\pi_\cU(F))$. Note that every subset of $\cD(\cA)$
which is uniquely determined by $\pi_\cU$ is lift-invariant under $\pi_\cU$. 
\par
\begin{Lem}\label{lem:UGS-UDA}
A subset $F\subseteq\cD(\cA)$ is the preimage $\pi_\cU|_{\cD(\cA)}^{-1}(G)$ of a face 
$G$ of $\pi_\cU(\cD(\cA))$ if and only if $F$ is a face of $\cD(\cA)$ which is
lift-invariant under $\pi_\cU$. If $F$ is a face of $\cD(\cA)$ lift-invariant under 
$\pi_\cU$, then $\pi_\cU(F)$ is an exposed face of $\pi_\cU(\cD(\cA))$ if and only 
if $F=\phi_\cA(P_0(A))$ holds for some $A\in\cU$.
\end{Lem}
\begin{proof}
The first statement is proved in Prop.~5.7 in~\cite{Weis2012} for faces $F$. 
This and the fact that preimages of faces are faces prove the statement for 
general subsets $F$. The second statement follows from 
Equation~\eqref{eq:isoP0FW-prop}.
\end{proof}
We remark that the inverse isomorphism $\phi_\cA^{-1}$, introduced in
Equation~\eqref{eq:iso-P-FD}, can be applied to the preimage 
$F=\pi_\cU|_{\cD(\cA)}^{-1}(G)$ of every face $G$ of $\pi_\cU(\cD(\cA))$, not 
only to the exposed faces. In the context of $\fg$-local Hamiltonians 
\eqref{eq:locHamilton}, the image $P(\C^d)\subseteq\C^d$ of the projector 
$P=\phi_\cA^{-1}(F)$ associated with $F$ has been called a 
\emph{$\fg$-correlated space} \cite{Chen-etal2012b}. 
\par
\begin{Cor}\label{cor:UGS-UDA}
For any state $\rho\in\cD(\cA)$, the singleton $\{\rho\}$ is the preimage of an 
extreme point of $\pi_\cU(\cD(\cA))$ if and only if $\rho$ is a pure state of 
$\cA$ which is 
uniquely determined by $\pi_\cU$. If $\rho$ is a pure state of $\cA$ 
uniquely determined by $\pi_\cU$, then $\pi_\cU(\rho)$ is an exposed point 
of $\pi_\cU(\cD(\cA))$ if and only if $\rho$ is the unique ground state of a 
matrix $A\in\cU$. 
\end{Cor}
\begin{proof}
The claim follows from Lemma~\ref{lem:UGS-UDA} as every lift-invariant singleton 
is uniquely determined by $\pi_\cU$.
\end{proof}
It was an open problem \cite{Chen-etal2012a,Chen-etal2012b} whether the set of 
quantum marginals can have nonexposed faces. Here, we discuss prior work 
\cite{Karuvade-etal2019,Chen-etal2012b} regarding two-body marginals of $N$ qubits. 
Using the notation of the Sections~\ref{sec:hierarchical} and~\ref{sec:first-glimpse}, 
we employ the algebra $\cA_i=M_2$ for each unit $i\in\Omega=\{1,2,\dots,N\}$. We 
denote the linear map \eqref{eq:marmap}, which assigns marginals, by
\[\textstyle
\mar_N=\mar_{({N\choose 2},\fa_{N,\mathrm{qu}})}
\]
and the space of two-local $N$-qubit Hamiltonians \eqref{eq:locHamilton} by
\[\textstyle
\cU_N=\cU({N\choose 2},\fa_{N,\mathrm{qu}}).
\]
Equation~\eqref{eq:mar-jnr} shows that the map $\,\mar_N$ factors through 
$\,\cU_N$ and that
\begin{equation}\label{eq:mar-jnr-qu}
\pi_{\cU_N}(\cD(\cA_\Omega))
\stackrel{\mar_N}{\longrightarrow}
\mar_N(\cD(\cA_\Omega))
\end{equation}
is a bijection from the joint numerical range onto the set of two-body 
marginals of $N$ qubits. 
The set $\,\mar_6(\cD(\cA_\Omega))$ has a nonexposed point:
\par
\begin{Rem}[Six Qubits]\label{rem:six-qubits}
Using dissipative quantum control theory,
Karuvade et al.~\cite[Section IV.B]{Karuvade-etal2019} discovered a pure six-qubit 
state $\rho\in\cD(\cA_\Omega)$ that is uniquely determined by its two-body 
marginals, but which is not the unique ground state of any matrix in $\cU_6$,
the space of two-local six-qubit Hamiltonians. Equation~\eqref{eq:mar-jnr-qu} shows 
that $\rho$ is uniquely determined by the projection $\pi_{\cU_6}$. Hence the point 
$\pi_{\cU_6}(\rho)$ is a nonexposed point of the joint numerical range 
$\pi_{\cU_6}(\cD(\cA_\Omega))$ by Corollary~\ref{cor:UGS-UDA}, and the point
$\,\mar_6(\rho)$ is a nonexposed point of the set $\,\mar_6(\cD(\cA_\Omega))$ 
of two-body marginals of six qubits, again by Equation~\eqref{eq:mar-jnr-qu}.
\end{Rem}
\begin{Lem}[Three Qubits]\label{lem:three-qubits}
The set $\,\mar_3(\cD(\cA_\Omega))$ of two-body marginals of three qubits has no 
nonexposed points.
\end{Lem}
\begin{proof}
Let $x$ be an extreme point of the convex set $\,\mar_3(\cD(\cA_\Omega))$. Then 
$x=\mar_3(\rho)$ is the image of an extreme point (pure state) $\rho$ of the state 
space $\cD(\cA)$. First, let $\rho$ be of the GHZ type 
\cite{BengtssonZyczkowski2017}
\[
\ket{\GHZ}=\alpha\ket{000}+\beta\ket{111}
\]
for some $\alpha,\beta\in\C$ satisfying $|\alpha|^2+|\beta|^2=1$. That is to say, 
\[
\rho=U_1U_2U_3.\ket{\GHZ}\!\!\bra{\GHZ}.U_1^\ast U_2^\ast U_3^\ast
\]
for some unitaries $U_1,U_2,U_3\in U(2)$. Applying a unitary similarity, we can 
take $U_i=I$, $i=1,2,3$, without loss of generality. Then the two-body marginals 
are $\mar_3(\rho)=y(|\beta|^2)$, where 
\[
y(\lambda)=(\sigma,\sigma,\sigma)
\quad\text{and}\quad
\sigma=(1-\lambda)\ket{00}\!\!\bra{00}+\lambda\ket{11}\!\!\bra{11},
\qquad 0\leq\lambda\leq 1.
\]
Since the segment $\{y(\lambda)\mid 0\leq\lambda\leq1\}$ lies in 
$\mar_3(\cD(\cA_\Omega))$ and since the point $x=y(|\beta|^2)$ is an extreme 
point, we have $\sigma=\ket{ii}\!\!\bra{ii}$ and $\rho=\ket{iii}\!\!\bra{iii}$ 
either for $i=0$ or $i=1$. This shows that $\rho$ is the unique ground state of 
the two-local Hamiltonian 
\[
(I_{\{1,2\}}-\sigma)I_{\{3\}}+I_{\{1\}}(I_{\{2,3\}}-\sigma).
\] 
Second, if the pure state $\rho$ is not of the GHZ type, then $\rho$ is the 
unique ground state of a two-local Hamiltonian, too 
\cite[Section~V.A]{Chen-etal2012b}. 
\par
In both cases, the state $\rho$ is the unique ground state of a two-local 
Hamiltonian. Equation~\eqref{eq:isoP0FW} then shows that $\pi_{\cU_3}(\rho)$ is 
an exposed point of the joint numerical range $\pi_{\cU_3}(\cD(\cA_\Omega))$.
Equation~\eqref{eq:mar-jnr-qu} proves that 
$\,\mar_3(\rho)$ is an exposed point of the set of marginals 
$\,\mar_3(\cD(\cA_\Omega))$.
\end{proof}
Remark~\ref{rem:six-qubits} and Lemma~\ref{lem:three-qubits} prompt the 
question of whether the convex sets $\,\mar_4(\cD(\cA_\Omega))$ and 
$\,\mar_5(\cD(\cA_\Omega))$ have nonexposed points. It would also be 
interesting to establish whether the convex set $\,\mar_3(\cD(\cA_\Omega))$ 
has nonexposed faces of higher dimensions $1,2,\dots,34$. 
\par
%
%
%
\vskip2\baselineskip
\noindent
{\footnotesize
Acknowledgements. We thank Tomasz Maci{\c a}{\.z}ek and Adam Sawicki 
for bringing the papers \cite{MaciazekTsanov2017,Yu-etal2021} and 
\cite{Schilling-etal2020}, respectively, to our attention.
This preprint has not undergone peer review or any post-submission 
improvements or corrections. The Version of Record of this article is 
published in Information Geometry, and is available online at 
\texttt{https://doi.org/10.1007/s41884-023-00103-2}.
A full-text view-only version is available online at 
\texttt{https://rdcu.be/c7b2s}.}
%
%
\bibliographystyle{plain}

\end{document}